\documentclass{LMCS}

\mathcode`:="003A  
\mathcode`;="003B  
\mathcode`?="003F  
\mathcode`|="026A  
\mathcode`<="4268  
\mathcode`>="5269  

\mathchardef\ls="213C    
\mathchardef\gr="213E    
\mathchardef\uparrow="0222  
\mathchardef\downarrow="0223  

\newcommand\D{\mathcal{D}}
\newcommand\N{\mathcal{N}}
\newcommand\M{\mathcal{M}}
\newcommand\Pa{\mathcal{Q}}
\def\pow{{\mathcal P_{\!\!\!\omega}}}
\newcommand\pto\rightharpoonup
\newcommand\E\varepsilon
\newcommand\Exp{\mathsf{Exp}}
\newcommand\Plus{\mathsf{Plus}}
\newcommand\Empty{\mathsf{Empty}}
\newcommand\ndf{\mathit{NDF}}
\newcommand\id{\mathsf{Id}}
\newcommand\B{\mathsf{B}}
\newcommand\G{\mathcal{G}}
\newcommand\F{\mathcal{F}}
\newcommand\Pol{\mathbf{P}}
\newcommand\Fin{\mathbf{F}}
\newcommand\emp{\underline\emptyset}

\newcommand{\rules}[2]{\mbox{$\frac%
                      {\mbox{\small \rule[-5pt]{0pt}{14pt} $#1$}}
                      {\mbox{\small \rule[0pt]{0pt}{10pt}$#2$}}$}}

\DeclareSymbolFont{lasy}{U}{lasy}{m}{n}
\DeclareMathSymbol\myDiamond{\mathord}{lasy}{"33}
\newcommand{\myplus}{\mathbin{\rlap{$\myDiamond$}\hspace*{.01cm}\raisebox{.14ex}{$+$}}}
\newcommand{\mybigplus}{\scalebox{1.5}{$\myplus$}}

\def\expr#1{<\!< \, #1 \, >\!>}
\def\ceil#1{\lceil\, #1 \,\rceil}

  \def\hyph{-\penalty0\hskip0pt\relax} 
\usepackage{verbatim}
\usepackage{longtable}
\usepackage{amsmath}
\usepackage{enumerate,hyperref}

\theoremstyle{definition}
\newtheorem{mydefinition}{\textsc{Definition}}[section]
\theoremstyle{plain}
\newtheorem{mylemma}[mydefinition]{\textsc{Lemma}}
\theoremstyle{plain}
\newtheorem{mytheorem}[mydefinition]{\textsc{Theorem}}
\theoremstyle{plain}
\newtheorem{myproposition}[mydefinition]{\textsc{Proposition}}
\theoremstyle{plain}

\theoremstyle{definition}

\theoremstyle{definition}
\newtheorem{myexample}[mydefinition]{\textsc{Example}}

\newenvironment{definition}{
\begin{mydefinition}}
    {\hfill$\clubsuit$\end{mydefinition}}

\newenvironment{example}{
\begin{myexample}}
    {\hfill$\spadesuit$\end{myexample}}

\newenvironment{lemma}{
\begin{mylemma}}
    {\end{mylemma}}

\newenvironment{theorem}{
\begin{mytheorem}}
    {\end{mytheorem}}

\usepackage{backref}
\usepackage{memhfixc}

\usepackage{color}
\usepackage{graphics}
\usepackage{graphicx}
\usepackage[all,2cell,cmtip,dvips,ps]{xy}

\usepackage{etex}
\usepackage{tikz}
\newcommand*\mycirc[1]{%
  \begin{tikzpicture}[baseline=(C.base)]
    \node[draw,circle,inner sep=1pt](C) {#1};
  \end{tikzpicture}}

\def\doi{6 (3:23) 2010}
\lmcsheading%
{\doi}
{1--39}
{}
{}
{Jan.~\phantom03, 2010}
{Sep.~\phantom09, 2010}
{}

\begin{document}

\title{Non-deterministic Kleene coalgebras}
\author[A.~Silva]{Alexandra Silva\rsuper a}  
\address{{\lsuper a}CWI, Amsterdam, The Netherlands}  
\email{ams@cwi.nl} 
\thanks{{\lsuper a}The first author was partially supported by the
Funda\c{c}\~ao para a Ci\^encia e a Tecnologia, Portugal, under grant
number SFRH/BD/27482/2006.} 

\author[M.~Bonsangue]{Marcello Bonsangue\rsuper b}
\address{{\lsuper b}LIACS, University of Leiden,The Netherlands}
\email{marcello@liacs.nl}  

\author[J.~Rutten]{Jan Rutten\rsuper c}	
\address{{\lsuper c}CWI (Amsterdam), VUA (Amsterdam) and RUN (Nijmegen) , The Netherlands}	
\email{janr@cwi.nl}  


\begin{abstract}
In this paper, we present a systematic way of deriving (1) languages of (generalised) regular expressions, and
(2) sound and complete axiomatizations thereof,
for a wide variety of systems. 
This generalizes both the
results of Kleene (on regular languages and deterministic finite
automata)
and Milner (on regular behaviours and finite labelled transition
systems),
and includes many other systems such as Mealy and Moore machines.
\end{abstract}

\maketitle

\section{Introduction}
In a previous paper~\cite{BRS08}, we presented a language to describe the 
behaviour of Mealy machines and a sound and complete axiomatization
thereof. The defined language and axiomatization can be seen as the analogue of classical regular
expressions~\cite{kleene} and Kleene algebra~\cite{kozen}, for
deterministic finite automata (DFA), or the
process algebra and axiomatization for labelled transition
systems (LTS)~\cite{milner}.

We now extend the previous approach and devise a framework wherein
languages and axiomatizations can be uniformly derived for a large
class of systems, including DFA, LTS and Mealy machines, which we
will model as coalgebras.
 
Coalgebras provide a general framework for the study of dynamical
systems such as DFA, Mealy machines and LTS. For a functor
$\G\colon  \mathbf{Set}
\to \mathbf{Set}$, a $\G$-coalgebra or $\G$-system is a pair $(S,g)$,
consisting of a set $S$ of states and a function $g\colon S\to \G(S)$
defining the ``transitions'' of the states. We call the functor $\G$
the \emph{type} of the system. For instance, DFA can be modelled as
coalgebras of the functor $\G(S) = 2 \times
S^A$, Mealy machines are obtained by taking $\G(S) = (B\times S)^A$ and 
image-finite LTS are coalgebras for the functor $\G(S) = (\pow
(S))^A$, where $\pow$ is finite powerset.

Under mild conditions, functors $\G$ have a \emph{final coalgebra}
(unique up to isomorphism) into which every $\G$-coalgebra can be
mapped via a unique so-called $\G$\hyph\emph{homomorphism}. The final
coalgebra can be viewed as the universe of all possible
$\G$-\emph{behaviours}: the unique homomorphism into the final
coalgebra maps every state of a coalgebra to a canonical
representative of its behaviour. This provides a general notion of
behavioural equivalence: two states are equivalent if and only if they are
mapped to the same element of the final coalgebra. Instantiating the
notion of final coalgebra for the aforementioned examples, the result
is as expected: for DFA the final coalgebra is the set $2^{A^*}$ of all languages over $A$; for Mealy
machines it is the set of causal functions $f\colon A^\omega\to B^\omega$;
and for LTS it is the set of
finitely branching trees with arcs labelled by $a\in A$ modulo bisimilarity. The notion of
equivalence also specializes to the familiar notions: for DFA, two states are
equivalent when they accept the same language;
for Mealy machines, if they realize (or compute) the same causal function; and 
for LTS if they are bisimilar.

It is the main aim of this paper to show how the type of a system,
given by the functor $\G$, is not only enough to determine a notion of
behaviour and behavioural equivalence, but also allows for a uniform
derivation of both 
a set of expressions describing behaviour and a corresponding
axiomatization. The theory of universal coalgebra~\cite{Rutten00}
provides a standard equivalence and a universal domain of behaviours,
uniquely based on the functor $\G$. The main contributions of this
paper are (1) the definition of a set of expressions $\Exp_\G$
describing $\G$-behaviours, (2) the proof of the correspondence between
behaviours described by $\Exp_\G$ and locally finite $\G$-coalgebras (this is the analogue of Kleene's theorem),  and (3) a corresponding sound and complete
axiomatization, with respect to bisimulation, of $\Exp_\G$ (this is the analogue of
Kleene algebra). All these results are solely based on the type of
the system, given by the functor $\G$. 

In a nutshell, we combine the work of Kleene
with coalgebra, considering the class of non-deterministic functors.
Hence, the title of the paper: non-deterministic Kleene coalgebras.   

\paragraph{\textbf{Organization of the paper.}} In Section~\ref{sec:ndc} we
introduce the class of non\hyph deterministic functors and coalgebras. In
Section~\ref{sec:expressions} we associate with each non\hyph
deterministic functor $\G$ a generalized language $\Exp_\G$ of regular
expressions and we present an analogue of Kleene's theorem, which
makes precise the connection between $\Exp_\G$  and $\G$-coalgebras. A
sound and complete axiomatization  of $\Exp_\G$ is presented in
Section~\ref{sec:axiom}. Section~\ref{sec:appl} contains two more
examples of application of the framework and Section~\ref{sec:pol_fin} shows a language
and axiomatization for the class of polynomial and finitary
coalgebras. Section~\ref{sec:conclusions5} presents
concluding remarks, directions for future work and discusses related work. 
This paper is an extended version of~\cite{regexp,BRS09b}: it includes
all the proofs, more examples and explanations, new material about polynomial and finitary functors and an extended discussion section.

\section{Preliminaries} \label{sec:ndc}
We give the basic definitions on non-deterministic functors and coalgebras
and introduce the notion of bisimulation.

\medskip\noindent
First we fix notation on sets and operations on them. Let
$\mathbf{Set}$ be the category of sets and functions. Sets are
denoted by capital letters $X,Y,\ldots$ and functions by lower case
$f,g,\ldots$. We write $\emptyset$ for the empty set and the collection of
all {\em finite} subsets of a set $X$ is defined as
$\pow(X) = \{Y\subseteq X\mid \mbox{\textrm{$Y$ finite}}\}$. The collection of functions from a set $X$ to a set
$Y$ is denoted by $Y^X$. We write $\mathit{id}_X$ for the identity
function on set $X$. Given functions $f\;:\; X\to Y$ and
$g\;:\;Y\to Z$ we write their composition as $g\circ f$. The product
of two sets $X,Y$ is written as $X\times Y$, with projection
functions
$
\xymatrix{X&X\times
Y\ar[l]_-{\pi_1}\ar[r]^-{\pi_2}&Y}.
$
 The set $1$ is a singleton set typically written as $1=\{ * \}$ and
it can be regarded as the empty product. We define $X\myplus Y$ as the set
$X\uplus Y \uplus \{\bot,\top\}$, where $\uplus$ is the disjoint
union of sets, with injections
$
\xymatrix{X\ar[r]^-{\kappa_1}&X
\uplus
Y &Y\ar[l]_-{\kappa_2}}$. 
 Note that the set $X\myplus Y$ is different from the classical
coproduct of X and Y (which we shall denote by $X+Y$), because of the two extra elements $\bot$ and
$\top$.
These extra elements will later be used to represent, respectively,
underspecification and inconsistency in the specification of some
systems. The intuition behind the
need of these extra elements will become clear when we present our language
of expressions and concrete examples, in
Section~\ref{sec:examples_synt}, of systems whose type
involves $\myplus$. Note that $X \myplus X
\not\cong 2\times X \cong X+ X$. 

For each of the operations defined above on sets, there are analogous
ones on functions. Let $f\colon X\to Y$, $f_1\colon X\to Y$ and $f_2\colon Z\to W$. We
define the following operations:
\[
\begin{array}{l@{\hspace{2.3cm}}l}
f_1\times f_2 \colon X\times Z \to Y\times W & f_1\myplus f_2 \colon X\myplus
Z \to Y\myplus W\\[.6ex]
(f_1\times f_2)(<x,z>) = <f_1(x),f_2(z)> &
(f_1\myplus f_2)(c) = c,\ c\in\{\bot,\top\}\\ &
(f_1\myplus f_2)(\kappa_i(x)) = \kappa_i(f_i(x)),\ i\in\{1,2\}\\[2.4ex]
f^A \colon X^A \to Y^A & \pow(f) \colon \pow (X) \to \pow (Y)\\[.6ex]

f^A (g) = f\circ g & \pow(f)(S) = \{f(x) \mid x\in S\}
\end{array}
\]
\noindent Note that here we are using the same symbols that we defined above for the
operations on sets. It will always be clear from the context which
operation is being used.

In our definition of non-deterministic functors we will use constant sets
equipped with an information order. In particular, we will use
join\hyph semilattices. A (bounded) join\hyph semilattice is a set $\B$ equipped
with a binary operation $\vee_\B$ and a constant $\bot_\B \in \B$, such
that $\vee_\B$ is commutative, associative and idempotent. The
element $\bot_\B$ is neutral with respect to $\vee_\B$. As usual, $\vee_\B$
gives rise to a partial ordering $\leq_\B$ on the elements of $\B$:
\[
b_1 \leq_\B b_2 \Leftrightarrow b_1\vee_\B b_2 = b_2
\]
Every set $S$
can be mapped into a join\hyph semilattice by taking $\B$ to be the
set of all finite subsets of $S$ with union as join.

\paragraph{\textbf{Non-deterministic functors.}} 
Non-deterministic functors 
are functors $\G\colon\textbf{Set} \to \textbf{Set}$, built inductively
from
the identity and constants, using $\times$, $\myplus$, $(-)^A$ and
$\pow$.
\begin{definition}\label{def:polfunctor}
The class $\ndf$ of {\em non\hyph deterministic functors}  on
$\mathbf{Set}$ is
inductively defined by putting:
\[
\ndf \ni \G ::= \id \mid \B \mid \G\myplus \G \mid \G\times
\G \mid \G^A \mid \pow \G
\]
where $\B$ is a finite (non-empty) join\hyph semilattice and $A$ is a finite set.
\end{definition}
Since we only consider finite exponents $A = \{a_1,\ldots,a_n\}$, the
functor $(-)^A$ is not really needed, since it is subsumed by a
product with $n$ components. However, to simplify the presentation, we
decided to
include it.

Next, we show the explicit definition of the functors above on a set $X$
and on a morphism $f\colon X\to Y$ (note that $\G(f) \colon \G(X) \to
\G(Y)$).
\[
\begin{array}{lll}
\id(X) = X & \B(X) = \B & (\G_1\myplus \G_2) (X) = \G_1(X)\myplus \G_2(X) \\
\id(f) = f &  \B(f) = \mathit{id}_\B & (\G_1\myplus \G_2) (f) = \G_1(f)
\myplus \G_2 (f) \\[1ex]
 (\G^A)(X) = \G(X)^A & (\pow
\G)(X)
= \pow (\G(X)) & (\G_1\times \G_2) (X) = \G_1(X) \times \G_2(X) \\
 (\G^A)(f) = \G(f)^A & 
(\pow \G)(f) = \pow (\G(f)) & (\G_1\times \G_2) (f) = \G_1(f) \times
\G_2(f)
\\
\end{array}
\]
Typical examples of non\hyph deterministic functors include $\M =
(\B\times \id)^A$,
 $\D = 2 \times \id^A$, $\Pa= (1 \myplus \id)^A$ and $\N= 2
\times (\pow\id)^A $, where $2 = \{0,1\}$ is a two-element join
semilattice with $0$ as bottom element ($1\vee 0 = 1$) and $1=\{*\}$ is a
one element join-semilattice. These
functors represent, respectively, the type  of Mealy,
deterministic, partial deterministic and non\hyph deterministic
automata. 
In this paper, we will use the last three as running examples.
In~\cite{BRS08}, we have studied in detail regular
expressions for Mealy automata. Similarly to what happened there, we
impose a join\hyph semilattice structure on the constant functor. The
product, exponentiation and powerset functors preserve the
join\hyph semilattice structure and thus do not need to be changed.
 This is not the case for the classical coproduct and
thus we use $\myplus$ instead, which also guarantees that the join
semilattice structure is preserved. 

Next, we give the definition of  the ingredient relation, which
relates a non\hyph deterministic functor $\G$ with its {\em
ingredients}, {\em
i.e.} the functors used in its inductive construction. We shall use
this relation later for typing our expressions.

\begin{definition}

Let $\lhd\subseteq \ndf\times \ndf$ be the least reflexive and
transitive relation on
non\hyph deterministic functors
such that
\[
\G_1\triangleleft \G_1\times \G_2,\ \ \ 
\G_2\lhd \G_1\times \G_2,\ \ \ 
\G_1\lhd \G_1\myplus \G_2,\ \ \ 
\G_2\lhd \G_1\myplus \G_2,\ \ \ 
\G\lhd \G^A,\ \ \ \G\lhd \pow \G
\]
\end{definition}
Here and throughout this document we use $\F \lhd \G$ as a shorthand
for $<\F,\G>\in \lhd$. If $\F \lhd \G$, then $\F$ is said to  be an
\emph{ingredient} of $\G$. For example, $2$, $\id$, $\id^A$ and $\D$
itself are all the ingredients of the deterministic automata functor
$\D=2\times\id^A$.

\paragraph{\textbf{Non-deterministic coalgebras.}} A {non\hyph
deterministic} coalgebra is a pair $(S, f\colon S\to \G(S))$, where
$S$ is a set of states and $\G$ is a non\hyph deterministic functor. 
 The functor $\G$, together with the
function $f$, determines the {\em transition structure} (or
dynamics) of the $\G$-coalgebra~\cite{Rutten00}. Mealy,
deterministic, partial deterministic and non\hyph deterministic
automata are, respectively, coalgebras for the functors
$\M =
(\B\times \id)^A$,
 $\D = 2 \times \id^A$, $\Pa= (1 \myplus \id)^A$ and $\N= 2
\times (\pow\id)^A $.

A {\em $\G$-homomorphism\/} from a $\G$-coalgebra $(S,f)$ to a
$\G$-coalgebra $(T,g)$ is a function $h\colon S \to T$ preserving the
transition structure, {\em i.e.} such that $g\circ h = \G(h) \circ
f$. 

\begin{definition}
A $\G$-coalgebra $(\Omega,\omega)$ is said to be {\em final} if
for any $\G$-coalgebra $(S,f)$ there exists a unique $\G$-homomorphism
 $\mathbf{beh}_S\colon S\to \Omega$. 
\end{definition}
For every non-deterministic functor $\G$
there exists a final $\G$-coalgebra $(\Omega_\G,
\omega_\G)$~\cite{Rutten00}. For instance, as we already mentioned in
the introduction, the final coalgebra for the functor $\D$ is the set
of languages $2^{A^*}$ over $A$, together with a transition function
$d\;:\;2^{A^*} \to 2\times(2^{A^*})^A $ defined as $d(\phi) =
<\phi(\epsilon), \lambda a\lambda w. \phi(aw)>$. Here $\epsilon$
denotes the empty sequence and $aw$ denotes the word resulting from
prefixing $w$ with the letter $a$. The notion of finality will play
a key role later in providing a semantics to expressions.

Given a $\G$-coalgebra $(S,f)$ and a subset $V$ of~$S$ with inclusion
map $i\colon V\to S$ we say that $V$ is a subcoalgebra of $S$ if there
exists $g\colon V\to \G(V)$ such that $i$ is a homomorphism.
Given $s\in
S$, $<s>=(T,t)$, denotes the smallest subcoalgebra generated by
$s$, with $T$ given by
\[
T = \bigcap \{V \mid V \text{ is a subcoalgebra of $S$ and } s\in V\}
\]
If the functor
$\F$ preserves arbitrary intersections, then the subcoalgebra $<s>$
exists. This will be the case for every functor considered in this
paper. Moreover, all the functors we will consider preserve
monos and thus the transition structure $t$ is
unique~\cite[Proposition 6.1]{Rutten00}. 
 
We will write $\mathit{Coalg}(\G)$ for the category of $\G$-coalgebras together
with coalgebra homomorphisms. We also write
$\mathit{Coalg}_{\mathbf{LF}}(\G)$ for the category of
$\G$-coalgebras that are {\em locally finite}. Objects are
$\G$-coalgebras
$(S,f)$ such that
for each state $s\in S$ the generated subcoalgebra $<s>$ is finite.
Maps are the usual homomorphisms of coalgebras.

Let $(S, f)$ and $(T,g)$ be two $\G$-coalgebras. We call a relation
$R \subseteq S\times T$ a {\em bisimulation\/}~\cite{HermidaJ98}
iff 
\[
<s,t>\in R \Rightarrow <f(s), g(t)>\in \overline \G(R)
\]
where $\overline \G(R)$ is defined as
 $
\overline \G(R) = \{ <\G(\pi_1)(x),\G(\pi_2)(x)> \mid x \in \G(R) \}
$. 
We write $s\sim_\G t$ whenever there exists a bisimulation relation
containing $(s,t)$ and we call $\sim_\G$ the bisimilarity relation.
We shall drop the subscript $\G$ whenever the functor $\G$ is clear
from the context. For all non-deterministic $\G$-coalgebras $(S,f)$ and $(T,g)$ and $s\in
S, t\in T$, it holds that $s\sim t \iff \mathbf{beh}_S(s) =
\mathbf{beh}_T(t)$ (the left to right implication always holds, whereas
the right to left implication only holds for certain classes of
functors, which include the ones we consider in this
paper~\cite{Rutten00,staton}). 

\section{A language of expressions for non\hyph deterministic
coalgebras}\label{sec:expressions}

In this section, we generalize the classical notion of regular
expressions to non\hyph deterministic coalgebras. We start by introducing an
untyped language of expressions and then we single out the
well-typed ones via an appropriate typing system, thereby associating
expressions to non\hyph deterministic functors. 

\begin{definition}[Expressions]
Let $A$ be a finite set, $\B$ a finite join\hyph semilattice and $X$ a set
of fixed point variables. The set $\Exp$ of all {\em expressions\/} is given
by the following grammar, where $a\in A$, $b\in \B$ and $x\in X$:
\[
\begin{array}{lcl}
\E &:: =& \emp \mid x \mid \E \oplus \E \mid \mu x.\gamma
    \mid b \mid l<\E> \mid r<\E> \mid l[\E] \mid r[\E] \mid a(\E) \mid
\{\E\}\\
\end{array}
\]
where $\gamma$ is a {\em guarded expression} given by:
\[
\begin{array}{lcl}
\gamma &:: =& \emp \mid \gamma \oplus \gamma \mid \mu
x.\gamma
    \mid b \mid l<\E> \mid r<\E> \mid l[\E] \mid r[\E] \mid a(\E) \mid
\{\E\}\\
\end{array}
\]
The only difference between the BNF of $\gamma$ and $\E$ is the
occurrence of $x$.
\end{definition}
 In the expression $\mu x.\gamma$, $\mu$ is a binder for
all the free occurrences of $x$ in $\gamma$. Variables that are not bound are
free. A {\em closed expression} is an expression
without free occurrences of fixed point variables $x$. We denote the set of closed
expressions by $\Exp^c$. 

Intuitively, expressions denote elements of the final coalgebra. The
expressions $\emp$, $\E_1\oplus \E_2$ and  $\mu x.\, \E$ will
play a similar role to, respectively, the empty language, the union
of languages and the Kleene star in classical regular expressions
for deterministic automata. The expressions $l<\E>$ and $r<\E>$ refer to the left and right hand-side
of products. Similarly, $l[\E]$ and $r[\E]$ refer to the left and right hand-side
of sums. The expressions $a(\E)$ and $\{\E\}$ denote function application and a singleton set, respectively. We shall soon
illustrate, by means of examples, the role of these expressions. Here,
it is already visible that our approach (to define a language) for the
powerset functor differs from classical modal logic where $\Box$ and
$\Diamond$ are used. 
 This is a choice, justified by the fact that our goal is to have a ``process algebra'' like 
language instead of a modal logic one. It also explains why we only consider finite powerset: 
every finite set can be written as the finite union of its singletons.

Our language does not have any operator denoting intersection or
complement (it only includes the sum operator $\oplus$). This is a
natural restriction, very much in the spirit of Kleene's regular
expressions for deterministic finite automata. We
will prove that this simple language is expressive enough to denote
exactly all locally finite coalgebras.

Next, we present a typing assignment system for associating
expressions to non\hyph deterministic functors. This will allow us to associate with each functor
$\G$ the expressions $\E\in \Exp^c$ that are valid specifications of
$\G$-coalgebras. The typing proceeds following
the structure of the expressions and the ingredients of the
functors.

\begin{definition}[Type system]\label{def:ts}
We define a typing relation $\vdash \subseteq \Exp \times \ndf
\times \ndf $ that will associate an expression $\E$ with two non\hyph
deterministic functors $\F$ and $\G$, which are related by the
ingredient relation ($\F$ is an ingredient of $\G$). We shall write
$\vdash \E\colon \F\lhd \G$ for $<\E,\F,\G> \in \;\vdash$.  The rules that define $\vdash$ are the following:
\[
\begin{array}{cccc}
\rules{}{\vdash \emp \colon \F\lhd \G }& 
\rules{}{\vdash b\colon \B\lhd \G}&
\rules{}{\vdash x \colon \G\lhd \G }&
\rules{\vdash\E\colon \G\lhd \G}
      {\vdash \mu x.\E \colon \G\lhd \G}\\\\
\rules{\vdash \E_1 \colon \F\lhd \G\;\;\;\; \vdash\E_2\colon \F\lhd
\G}{\vdash \E_1\oplus\E_2 \colon \F\lhd \G} &
\rules{\vdash \E \colon \G\lhd \G}
      {\vdash \E \colon \id\lhd \G}&
\rules{\vdash \E \colon \F \lhd \G}{\vdash \{\E\} \colon \pow \F \lhd \G}&
\rules{\vdash \E\colon \F\lhd \G}
      {\vdash a(\E) \colon \F^A\lhd \G}
\\\\
\rules{\vdash \E\colon \F_1\lhd \G}
      {\vdash l<\E> \colon \F_1\times \F_2\lhd \G}&
\rules{\vdash \E\colon \F_2\lhd \G}
      {\vdash r<\E> \colon \F_1\times \F_2\lhd \G}&
\rules{\vdash \E\colon \F_1\lhd \G}
      {\vdash l[\E] \colon \F_1\myplus \F_2\lhd \G}&
\rules{\vdash \E\colon \F_2\lhd \G}
      {\vdash r[\E] \colon \F_1\myplus \F_2\lhd \G}\end{array}
\]
\end{definition}
Intuitively, $\vdash \E\colon \F\lhd
\G$ (for a closed expression $\E$) means that $\E$ denotes an element 
 of $\F(\Omega_\G)$, where $\Omega_\G$ is the final
coalgebra of $\G$. As expected, there is a rule for each expression
construct. The extra rule involving $\id\lhd \G$ reflects
the isomorphism between the final coalgebra $\Omega_\G$ and
$\G(\Omega_\G)$ (Lambek's lemma, cf.~\cite{Rutten00}). Only fixed points at the outermost level of the
functor are allowed. This does not mean however that we disallow nested
fixed points. For instance, $\mu x.\, a(x \oplus \mu y.\, a(y))$ would be
a well-typed expression for the functor $\D$ of deterministic automata,
as it will become clear below, when we will present more examples of well-typed and non-well-typed
expressions. The presented type system is decidable
(expressions are of
finite length and the system is inductive on the structure of $\E\in
\Exp$). Note that the rules above are meant to be read as an inductive
definition rather than as an algorithm. In an eventual implementation,
extra care is needed in the case $\G=\id$, to avoid looping in the rule for
$\id\lhd \G$.  

We can formally define the set of $\G$-expressions: (closed and
guarded) well-typed
expressions associated with a non\hyph deterministic functor $\G$.

\begin{definition}[$\G$-expressions]
Let $\G$ be a non\hyph deterministic functor and $\F$ an ingredient of
$\G$.
We define $\Exp_{\F\lhd \G}$ by:
\[
\Exp_{\F\lhd \G} = \{\E \in \Exp^c \mid\ \vdash\E \colon \F\lhd \G\}\,.
\]
We define the set $\Exp_\G$ of well-typed {\em
$\G$-expressions\/} by $\Exp_{\G\lhd \G}$.
\end{definition}
Let us instantiate the definition of $\G$-expressions to the functors of
deterministic automata $\D = 2\times \id^A$.
\begin{example}[Deterministic expressions]
Let $A$ be a finite set of input actions and let $X$ be a set of
(recursion or) fixed point variables. The set $\Exp_\D$ of {\em deterministic
expressions\/} is given by the set of closed and guarded (each
variable occurs in the scope of $a(-)$) expressions generated by the following BNF grammar. For $a \in A$ and $x \in X$:
\[
\begin{array} {l@{\;}l@{\;}c@{\;}l}
\Exp_\D\ni &\E &::=&  \emp \mid \E \oplus \E \mid 
         \mu x.\E \mid 
 x \mid l<\E_1> \mid r<\E_2> \\
&\E_1 &::=& \emp \mid 0 \mid 1 \mid \E_1\oplus\E_1\\
&\E_2 &::=&  \emp \mid a(\E) \mid \E_2\oplus\E_2
\end{array}
        \]
\end{example}
Examples of well-typed expressions for the functor $\D = 2\times \id^A$
(with $2=\{0,1\}$ a two-element join-semilattice with $0$ as bottom
element; recall that the ingredients of $\D$ are $2$, $\id^A$ and $\D$
itself) include
$r<a(\emp)>$, $l<1>\oplus r<a(l<0>)>$ and $\mu x. r<a(x)> \oplus
l<1>$.
The expressions $l[1]$, $l<1>\oplus 1$ and $\mu x. 1$ are examples of
non well-typed
expressions for $\D$, because the functor $\D$ does not involve $\myplus$, the
subexpressions in the sum have different type, and
recursion is not at the outermost level ($1$ has type $2\lhd \D$),
respectively.

It is easy to see that the closed (and guarded) expressions generated
by the grammar presented above are exactly the elements of
$\Exp_\D$. The most interesting case to check is the expression
$r<a(\E)>$. Note that $a(\E)$ has type
$\id^A\lhd \D$ as long as $\E$ has type $\id \lhd \D$. And the crucial
remark here is that, by definition of $\vdash$, $\Exp_{\id\lhd
\G} \subseteq \Exp_\G$.  
Therefore, $\E$ has type $\id \lhd \D$ if it is of type
$\D\lhd \D$, or more precisely, if $\E\in \Exp_\D$, which explains why the
grammar above is correct. 

At this point, we should remark that the syntax of our expressions
differs from the classical regular expressions in the use of $\mu$
and action prefixing $a(\E)$ instead of star and full concatenation.
We shall prove later that these two syntactically different
formalisms are equally expressive (Theorems~\ref{kleene1} and
~\ref{kleene2}), but, to increase the intuition behind our expressions,
 let us present the syntactic translation from classical
regular expressions to $\Exp_\D$ (this translation is inspired
by~\cite{milner}) and back. 

\begin{definition}\label{def:regexp_to_exp} The set of regular expressions is given by the following syntax
\[
RE \ni r ::= \underline 0 \mid \underline 1 \mid a \mid r+r \mid r\cdot r \mid r^*
\]
where $a\in A$ and $\cdot$ denotes sequential composition. We define the following translations between regular expressions and deterministic expressions:
$$
\begin{array}{llp{.8cm}ll}
\multicolumn{2}{l}{(-)^\dagger\colon RE \to
\Exp_\D}&&\multicolumn{2}{l}{(-)^\ddagger\colon \Exp_\D \to RE}\\
 (\underline 0)^\dagger &= \emp &&  (\emp)^\ddagger
&= \underline 0\\
 (\underline 1)^\dagger &= l<1> && (l<\emp>)^\ddagger &=
(l<0>)^\ddagger = (r<\emp>)^\ddagger = \underline 0 \\
  (a)^\dagger &= r<a(l<1>)> && (l<1>)^\ddagger &=\underline 1 \\
 (r_1+r_2)^\dagger &= (r_1)^\dagger\oplus (r_2)^\dagger &&
(l<\E_1\oplus \E_2>)^\ddagger &=(l<\E_1>)^\ddagger+ (l<\E_2>)^\ddagger\\
 (r_1\cdot r_2)^\dagger &= (r_1)^\dagger [(r_2)^\dagger/l<1>] &&
(r<a(\E)>)^\ddagger &=a \cdot (\E)^\ddagger\\
  (r^*)^\dagger &= \mu x. (r)^\dagger[x/l<1>]\oplus l<1> &&
(r<\E_1\oplus \E_2>)^\ddagger &=(r<\E_1>)^\ddagger+
(r<\E_2>)^\ddagger\\
&&& (\E_1\oplus\E_2)^\ddagger &= (\E_1)^\ddagger+ (\E_2)^\ddagger\\
&&&  (\mu x. \E)^\ddagger &= \mathbf{sol}(\mathbf{eqs}(\mu x. \E)) 
\end{array}
$$
The function $\mathbf{eqs}$ translates $\mu x. \E$ into a system of
equations in the following way. Let $\mu x_1. \E_1, \ldots, \mu x_n.
\E_n$  be all the fixed point subexpressions of $\mu x. \E$, with
$x_1=x$ and $\E_1=\E$. We define $n$ equations $x_i =
(\overline\E_i)^\dagger$, where $\overline\E_i$ is obtained from
$\E_i$ by replacing each subexpression $\mu x_i. \E_i$ by $x_i$, for
all $i=1,\ldots n$. The solution of the system,
$\mathbf{sol}(\mathbf{eqs}(\mu x.\E))$, is then computed in the usual way (the solution of an equation of shape $x=rx+t$ is $r^*t$).

In~\cite{Rut98c}, regular expressions were given a coalgebraic
structure, using Brzozowski derivatives~\cite{Brz64}. Later in
this paper, we will provide a coalgebra structure to $\Exp_\D$, after
which the soundness of the above translations can be stated and proved: $r \sim
r^\dagger$ and $\E \sim \E^\ddagger$, where $\sim$ will coincide with language equivalence.
\end{definition}
Thus, the regular expression $aa^*$  is translated to $r<a(\mu x. r<a(x)> \oplus l<1>)>$, 
whereas the expression $\mu x. r<a(r<a(x)>)> \oplus l<1>$ is transformed into $(aa)^*$.

We present next the syntax for the
expressions in $\Exp_\Pa$ and in $\Exp_\N$ (recall that
$\Pa=(1\myplus\id)^A$ and $\N=2\times (\pow\id)^A$).

\begin{example}[Partial expressions]
Let $A$ be a finite set of input actions and $X$ be a set of (recursion or)
fixed point variables. The set $\Exp_\Pa$ of {\em partial expressions\/} is given by the set of closed and guarded expressions generated by the following BNF grammar. For $a \in A$ and $x \in X$:
\[
\begin{array} {l@{\;}l@{\;}c@{\;}l}
\Exp_\Pa\ni &\E &::=& \emp \mid  \E \oplus \E \mid 
         \mu x.\E \mid  x \mid 
         a(\E_1)\\
& \E_1&::=& \emp \mid \E_1\oplus \E_1   \mid l[\E_2] \mid r[\E] \\   
& \E_2 &::=& \emp \mid \E_2\oplus \E_2 \mid * 
\end{array}
\]
Intuitively, the expressions $a(l[*])$ and $a(r[\E])$ specify,
respectively, a state which has no
defined transition for input $a$ and a state with an outgoing
transition to another one specified by $\E$.  
\end{example}
\begin{example}[Non-deterministic expressions]
Let $A$ be a finite set of input actions and $X$ be a set of (recursion or)
fixed point variables. The set $\Exp_\N$ of {\em non\hyph deterministic expressions\/} is given by the set of closed and guarded expressions generated by the following BNF grammar. For $a \in A$ and $x \in X$:
\[
\begin{array} {l@{\;}l@{\;}c@{\;}l}
\Exp_\N \ni& \E &::=&  \emp \mid x \mid 
         r<\E_2> \mid  l<\E_1> \mid 
         \E \oplus \E \mid 
         \mu x.\E\\
&\E_1 &::=&  \emp \mid \E_1 \oplus \E_1 \mid 1 \mid 0\\
&\E_2 &::=&  \emp \mid \E_2 \oplus \E_2 \mid a(\E')\\
&\E' &::=&  \emp \mid \E' \oplus \E' \mid \{\E\}
\end{array}
\]
Intuitively, the expression $r<a(\{\E_1\}\oplus \{\E_2\})>$ specifies a
state which has two outgoing transitions labelled with the input letter $a$, one to a
state specified by $\E_1$ and another to a state specified by $\E_2$.
\end{example}

We have defined a language of expressions which gives us an
algebraic description of systems.  We should also remark at this point 
that in the examples we strictly follow the type system to derive the 
syntax of the expressions. However, it is obvious that many 
simplifications can be made in
order to obtain a more polished language. In particular, after the
axiomatization we will be able to decrease the number of levels in the
above grammars, since will we have axioms of the shape $a(\E)\oplus
a(\E') \equiv a(\E\oplus \E')$. In Section~\ref{sec:appl}, we will
sketch two examples where we apply some simplification to the syntax. 

The goal is now to present a generalization 
of Kleene's theorem for non\hyph deterministic coalgebras (Theorems~\ref{kleene1} and \ref{kleene2}). Recall that, for regular languages, the theorem states that a language is regular
if and only if it is recognized by a finite automaton. In order to
achieve our goal we will first show that the set $\Exp_\G$ of
$\G$-expressions carries a $\G$-coalgebra structure. 

\subsection{Expressions are coalgebras}\label{sec:expressions_coalg}

In this section, we show that the set of $\G$-expressions for a given
non\hyph deterministic functor $\G$ has a coalgebraic structure
$\delta_\G \colon
\Exp_\G\to \G(\Exp_\G)\,$. More precisely, we
are going to define a function
\[
\delta_{\F \lhd \G}\; : \; \Exp_{\F\lhd \G} \to \F (\Exp_\G)
\]
for every ingredient $\F$ of $\G$, and then set $\delta_\G =
\delta_{\G\lhd \G}$. Our definition of the
function $\delta_{\F \lhd \G}$ will make use of the following.
\begin{definition}For every $\G\in \ndf$ and for every $\F$ with
$\F\lhd \G$:
\begin{itemize} 
\item[(i)] we define a constant $\Empty_{\F\lhd \G} \in
\F(\Exp_{\G})$ by induction on the syntactic structure of $\F$:
\[
\begin{array}{lcl}
\Empty_{\id\lhd \G} &=& \emp \\
\Empty_{\B\lhd \G} &=& \bot_\B\\
\Empty_{\F_1\times \F_2\lhd \G} &=& <\Empty_{\F_1\lhd
\G}, \Empty_{\F_2\lhd \G}>
\end{array}
\begin{array}{lcl}
\Empty_{\F_1\myplus \F_2\lhd \G} &=& \bot\\
\Empty_{\F^A\lhd \G} &=& \lambda a. \Empty_{\F\lhd \G}\\
\Empty_{\pow \F\lhd \G} &=& \emptyset
\end{array}
\]
\item[(ii)] we define  a function $\Plus_{\F\lhd \G}\colon
\F(\Exp_{\G})\times
\F(\Exp_{\G})\to
\F(\Exp_\G)$ by induction on the syntactic structure of $\F$:
\[
\begin{array}{lcl}
\Plus_{\id\lhd \G} (\E_1, \E_2) &=& \E_1 \oplus \E_2\\
\Plus_{\B\lhd \G} (b_1,b_2) &=& b_1\vee_\B b_2\\
\Plus_{\F_1\times \F_2 \lhd \G} (<\E_1, \E_2>, <\E_3,
\E_4>)& =& <\Plus_{\F_1\lhd \G} (\E_1,\E_3),
\Plus_{\F_2\lhd \G} (
\E_2, \E_4)>\\
\Plus_{\F_1 \myplus \F_2\lhd \G} (\kappa_i(\E_1),\kappa_i(\E_2)) &=&
\kappa_i(\Plus_{\F_i\lhd \G} (\E_1,\E_2)),\ \ \ i\in \{1,2\}\\
\Plus_{\F_1 \myplus \F_2\lhd \G} (\kappa_i(\E_1),\kappa_j(\E_2)) &=& \top
\ \ \ i,j\in\{1,2\}\ \ and\ \ i\neq j\\
\Plus_{\F_1 \myplus \F_2\lhd \G} (x,\top) &=& \Plus_{\F_1
\myplus \F_2\lhd \G}
(\top,x) = \top\\
\Plus_{\F_1 \myplus \F_2\lhd \G} (x,\bot) &=& \Plus_{\F_1
\myplus \F_2\lhd \G}
(\bot,x) = x\\
\Plus_{\F^A\lhd \G}(f,g)& =& \lambda a.\ \Plus_{\F\lhd
\G}(f(a),g(a))\\
\Plus_{\pow \F\lhd \G}(s_1,s_2)&=& s_1\cup s_2
\end{array}
\]
\end{itemize}
Intuitively, one can think of the constant $\Empty_{\F \lhd \G}$
and the function $\Plus_{\F\lhd \G}$ as liftings of $\emp$
and $\oplus$ to the level of $\F(\Exp_{\G}$).
\end{definition}
We need two more things to define $\delta_{\F \lhd \G}$. First, we define an order $\preceq$ on the types of expressions. 
For $\F_1$, $\F_2$ and $\G$ non-deterministic functors such that
$\F_1\lhd \G$ and $\F_2\lhd \G$, we define $$(\F_1\lhd \G) \preceq
(\F_2\lhd \G) \Leftrightarrow \F_1\lhd \F_2$$
The order $\preceq$ is a partial order (structure inherited from
$\lhd$). Note also that  $(\F_1\lhd \G) = (\F_2\lhd \G) \Leftrightarrow \F_1=\F_2$. Second, we define a measure $N(\E)$ based on the maximum number of nested unguarded occurrences of $\mu$-expressions in $\E$ and unguarded occurrences of $\oplus$.  We say that a subexpression $\mu x.\E_1$ of $\E$  occurs unguarded if it is not in the scope of one of the operators $l<->$, $r<->$, $l[-]$, $r[-]$, $a(-)$ or $\{-\}$. 
\begin{definition}\label{def:N} For every guarded expression $\E$, we define $N(\E)$ as follows:
\begin{align*}
&N(\emp) = N(b) = N(a(\E)) = N(l<\E>) = N(r<\E>) = N(l[\E]) =
N(r[\E]) = N(\{\E\})= 0
\\
&N(\E_1 \oplus \E_2) = 1 + \mathit{max} \{ N(\E_1) , \, N(\E_2) \}\\
& N(\mu x.\E) = 1 + N(\E)
\end{align*}
\end{definition}
\noindent The measure $N$ induces a partial order on the set of expressions: $\E_1\ll \E_2 \Leftrightarrow N(\E_1) \leq N(\E_2)$, where $\leq$ is just the ordinary inequality of natural numbers. 

Now we have all we need to define $\delta_{\F \lhd \G}\colon
\Exp_{\F\lhd \G}\to \F(\Exp_\G)$.
\begin{definition}
For every ingredient $\F$ of a non\hyph deterministic functor
$\G$ and an expression $\E\in \Exp_{\F\lhd \G}$, we define $\delta_{\F\lhd
\G}(\E)$ as follows:
\[
\begin{array}{lcl}
\delta_{\F\lhd \G}  (\emp) &=& \Empty_{\F\lhd \G}\\
\delta_{\F\lhd \G}  (\E_1\oplus\E_2) &=& \Plus_{\F\lhd \G}(\delta_{\F\lhd
\G}  (\E_1), \delta_{\F\lhd \G}  (\E_2))\\
\delta_{\G\lhd \G}  (\mu x. \E) &=& \delta_{\G\lhd \G} (\E[\mu x.
\E/x])\\
\delta_{\id \lhd \G}(\E) &=& \E\ \ \ \textrm{for} \ \G\neq \id\\
\delta_{\B \lhd \G} (b) &=& b\\
\delta_{\F_1\times \F_2 \lhd \G} (l<\E>) &=& <
\delta_{\F_1\lhd \G}(\E), \Empty_{\F_2\lhd \G} >\\
\delta_{\F_1\times \F_2 \lhd \G} (r<\E>) &=& <\Empty_{\F_1\lhd \G},
\delta_{\F_2\lhd \G}(\E) >\\
\delta_{\F_1\myplus \F_2 \lhd \G} (l[\E]) &=& \kappa_1(\delta_{\F_1
\lhd \G}(\E))\\
\delta_{\F_1\myplus \F_2 \lhd \G} (r[\E]) &=& \kappa_2(\delta_{\F_2
\lhd \G}(\E))\\
\delta_{\F^A \lhd \G} (a(\E)) &=& \lambda a' .\left\{
\begin{array}{ll}\delta_{\F \lhd \G} (\E)
&\text{if }a=a'\\\Empty_{\F\lhd \G}&\text{otherwise}\end{array}\right.
\\
\delta_{\pow \F\lhd \G} (\{\E\}) &=& \{\,\delta_{\F
\lhd \G}(\E)\,\}
\end{array}
\]
Here, $\E[\mu x.\E/x]$ denotes syntactic substitution,
replacing every free occurrence of $x$  in $\E$ by $\mu x.\E$.
\end{definition}
In order to see that the definition of $\delta_{\F\lhd \G}$ is
well-formed, we have to observe that $\delta_{\F\lhd \G}$ can be seen
as a function having two arguments: the type $\F\lhd \G$ and the
expression $\E$. Then, we use induction on the Cartesian product of
types and expressions with orders $\preceq$ and $\ll$, respectively.
More precisely, given two pairs $<\F_1\lhd \G, \E_1>$ and  $<\F_2\lhd \G, \E_2>$ we have an order 
\begin{equation}\label{ind:order5}
\begin{array}{lcll}
<\F_1\lhd \G, \E_1> \leq <\F_2\lhd \G, \E_2> &\Leftrightarrow&&
\text{(i) } (\F_1\lhd \G) \preceq (\F_2\lhd \G) \\ &&\text{ or }&
\text{(ii) } (\F_1\lhd \G) = (\F_2\lhd \G)\text{ and } \E_1\ll \E_2
\end{array}
\end{equation}
Observe that in the definition above it is always true that
$<\F'\lhd \G, \E'> \leq <\F\lhd \G, \E>$, for all occurrences of
$\delta_{\F'\lhd \G}(\E')$ occurring in the right hand side of the
equation defining $\delta_{\F\lhd \G}(\E)$. In all cases, but the ones
that $\E$ is a fixed point or a sum expression, the inequality comes
from point (i) above. For the case of the sum, note that $<\F\lhd
\G,\E_1> \leq <\F\lhd \G,  \E_1\oplus\E_2>$ and  $<\F\lhd \G, \E_2> \leq
<\F\lhd \G,  \E_1\oplus\E_2>$ by point (ii), since $N(\E_1) \ls  N(\E_1\oplus\E_2)$ and $N(\E_2) \ls N(\E_1\oplus\E_2)$. Similarly, in the case of $\mu x.\E$ we
have that $N(\E) = N(\E[\mu x.\E/x])$, which can easily be proved by (standard) induction on the
syntactic structure of $\E$, since $\E$ is guarded (in $x$), and this
guarantees that $N(\E[\mu x.\E/x]) \ls N(\mu x. \E)$. Hence, $<\G\lhd
\G,\E> \leq <\G\lhd \G,\mu x.\E>$. Also note that clause
4 of the above definition overlaps with clauses 1 and 2 (by taking
$\F=\id$). However, they give the same result and
thus the function $\delta_{\F\lhd \G}$ is well-defined.

\begin{definition}
We define, for each non\hyph deterministic functor $\G$, a $\G$-coalgebra
\[
\delta_\G\colon \Exp_\G \to \G(\Exp_\G)
\]
by putting $\delta_\G = \delta_{\G\lhd \G}$.
\end{definition}
The function $\delta_\G$ can be thought of as the generalization of
the well-known notion of Brzozowski derivative~\cite{Brz64} for regular
expressions and, moreover, it provides an operational semantics for
expressions, as we shall see in Section~\ref{sec:expressive}.

The observation that the set of expressions has a coalgebra structure
will be crucial for the proof of the generalized Kleene theorem,
as will be shown in the next two sections.

\subsection{Expressions are expressive}\label{sec:expressive}

Having a $\G$-coalgebra structure on $\Exp_\G$ has two advantages. First,
it provides us, by finality, directly with a natural semantics
because of the existence of a (unique) homomorphism $\mathbf{beh} \colon \Exp_\G
\to \Omega_\G $, that
 assigns to every expression $\E$ an element $\mathbf{beh}(\E)$ of the
final coalgebra $\Omega_\G$.

The second advantage of the coalgebra structure on $\Exp_\G$ is that
it lets us use the notion of $\G$-bisimulation to relate
$\G$-coalgebras $(S,g)$ and expressions $\E\in \Exp_\G$. If one can
construct a bisimulation relation between an expression $\E$ and a
state $s$ of a given coalgebra, then the behaviour represented by
$\E$ is equal to the behaviour of the state $s$. This is the
analogue of computing the language $L(r)$ represented by a given
regular expression $r$ and the language $L(s)$ accepted by a state
$s$ of a finite state automaton and checking whether $L(r) = L(s)$.

The following theorem states that every state in a locally finite
$\G$-coalgebra can be represented by an expression in our language.
This generalizes {\em half} of Kleene's theorem for deterministic automata: if a language is
accepted by a finite automaton then it is regular ({\em i.e.} it can
be denoted by a regular expression). The
generalization of the other {\em half} of the theorem (if a language
is regular then it is accepted by a finite automaton) will be
presented in Section~\ref{sec:synthesis}. It is worth to remark that
in the usual definition of deterministic automaton the initial state
of the automaton is included and, thus, in the original Kleene's
theorem, it was enough to consider finite automata. In the
coalgebraic approach, the initial state is not explicitly modelled and
thus we need to consider locally-finite coalgebras: coalgebras where each state will generate a finite subcoalgebra.
\begin{theorem}\label{kleene1}
Let $\G$ be a non\hyph deterministic functor and let $(S,g)$ be a
locally-finite $\G$-coalgebra. Then, for any $s \in S$, there exists an
expression $\expr s \in \Exp_\G$ such that $s\sim \expr s$. 
\end{theorem}

\proof 
Let $s\in S$ and let $<s>= \{s_1, \ldots, s_n\}$ with $s_1=s$.
We construct, for every state $s_i \in <s>$, an expression $\expr {s_i} $ such that $s_i \sim \expr {s_i} $ .

If $\G= \id$, we set, for every $i$, $\expr {s_i} = \emp$.  It is easy to see that $\{<s_i,\emp> \mid s_i\in <s> \}$ is a bisimulation and, thus, we have that $s \sim \expr s$.

For $\G\neq \id$, we proceed in the following way. Let, for every $i$,
$A_i = \mu x_i. \gamma^\G_{g(s_i)}$ where, for $\F\lhd \G$ and $c\in \F
<s>$, the expression $\gamma^\F_{c}\in \Exp_{\F\lhd \G}$ is defined by
induction on the structure of~$\F$:
\[
\begin{array}{llll}
\gamma^{\id}_{s_i} = x_i & \gamma^{\B}_b = b &
{\gamma^{\F_1\times \F_2}_{<c,c'>} = l<\gamma^{\F_1}_{c}> \oplus
r<\gamma^{\F_2}_{c'}>} & \gamma^{\F^A}_f =  \bigoplus\limits_{a\in A}
a(\gamma^\F_{f(a)})
\\[2ex]
\gamma^{\F_1\myplus \F_2}_{\kappa_1(c)} = l[\gamma^{\F_1}_{c}] & 
\gamma^{\F_1\myplus \F_2}_{\kappa_2(c)} = r[\gamma^{\F_2}_{c}] &
\gamma^{\F_1\myplus \F_2}_{\bot} = \emp & 
\gamma^{\F_1\myplus \F_2}_{\top} = l[\emp]\oplus r[\emp] \\[2ex]
\multicolumn{2}{l}{\gamma^{\pow \F}_{C} = \begin{cases}
\bigoplus\limits_{c\in C}
\{\gamma^\F_{c}\} & C\neq\emptyset\\\emp &\text{otherwise}\end{cases} }
\end{array}
\]
Note that here the choice of $l[\emp]\oplus
r[\emp]$ to 
represent inconsistency is arbitrary but {\em canonical}, in the
sense that any other expression involving sum of $l[\E_1]$ and
$r[\E_2]$ will be bisimilar. Formally, the definition of $\gamma$
above is parametrized by a function from $\{s_1,\ldots, s_n\}$ to a
fixed set of
variables $\{x_1,\ldots, x_n\}$. It should also be noted that 
$\bigoplus\limits_{i\in I} \E_i$ stands for $\E_1 \oplus (\E_2 \oplus
(\E_3 \oplus
\ldots))$ (this is a choice, since later we will axiomatize $\oplus$
to be commutative and associative).

Let $A^0_i = A_i$, define $A_i^{k+1} =
A_i^k\{A^k_{k+1}/x_{k+1}\}$ and then set  $\expr{s_i} = A_i^n$. Here,
$A\{A'/x\}$ denotes syntactic replacement (that is, substitution
without renaming of bound variables in $A$ which are also free
variables in $A'$). The definition of $\expr{s_i}$ does not depend in
the chosen order of $\{s_1,\dots, s_n\}$: the expressions obtained are
just different modulo renaming of variables. 

Observe that  the term 
\[
A^n_i = ( \mu x_i.  \gamma^\G_{g(s_i)}) \{A^0_{1}/x_{1}\} \ldots \{A^{n-1}_{n}/x_{n}\} 
\]
is a closed term because, for every $j = 1,\ldots, n$, the term $A_j^{j-1}$ contains at most $n-j$ free variables in the set $\{x_{j+1}, \ldots, x_n\}$.

It remains to prove that $s_i\sim \expr {s_i}$. We show that 
 \mbox{$
R= \{ <s_i, \expr{s_i} >\mid s_i\in <s>\} $ }
is a bisimulation. For that, we define, for $\F\lhd \G$ and $c\in
\F<s>$, $\xi^\F_{c} = \gamma^\F_{c}  \{A^0_{1}/x_{1}\} \ldots \{A^{n-1}_{n}/x_{n}\}$ and the relation
\[
R_{\F\lhd \G} = \{<c, \delta_{\F\lhd \G}(\xi^\F_{c})> \mid c \in \F<s> \}.
\]
Then, we  prove that \mycirc{$1$} $R_{\F\lhd \G} = \overline \F(R)$ and \mycirc{$2$} $<g(s_i), \delta_\G(\expr{s_i})> \in R_{\G\lhd \G}$.

\begin{itemize}
\item[\mycirc{$1$}] By induction on the structure of $\F$.

 \fboxsep=2pt
\fbox{$\F=\id$} Note that $R_{\id\lhd \G} = \{<s_i, \xi^{\id}_{s_i}> \mid s_i\in <s>\}$  which is equal to $\id(R)=R$ provided that $\xi^{\id}_{s_i} = \expr{s_i}$. The latter is indeed the case:
\[
\begin{array}{lcll}
\xi^{\id}_{s_i} &=& \gamma^{\id}_{s_i}  \{A^0_{1}/x_{1}\} \ldots \{A^{n-1}_{n}/x_{n}\} & \text{(def. $\xi^{\id}_{s_i}$)}\\[1ex]
&=& x_i  \{A^0_{1}/x_{1}\} \ldots \{A^{n-1}_{n}/x_{n}\} & \text{(def. $\gamma^{\id}_{s_i}$)}\\[1ex]
&=& A^{i-1}_i \{A^{i}_{i+1}/x_{i+1}\} \ldots \{A^{n-1}_{n}/x_{n}\} & \text{($\{A_i^{i-1}/x_i\}$)} \\[1ex]
&=& A^0_i \{A^0_{1}/x_{1}\} \ldots \{A^{n-1}_{n}/x_{n}\}& \text{(def. $A^{i-1}_i$)}\\[1ex]
&=& \expr{s_i}& \text{(def. $\expr{s_i}$)}
\end{array} 
\]

\fbox{$\F=\B$} Note that, for $b\in \B$, $\xi^\B_b = \gamma^\B_b
\{A^0_{1}/x_{1}\} \ldots \{A^{n-1}_{n}/x_{n}\} = b$. Thus, we have
that $R_{\B\lhd \G} = \{<s_i, \xi^\B_{s_i}> \mid s_i\in \B <s>\} =  \{<b, b> \mid b\in \B\}  = \overline \B(R)$.  

\fbox{$\F=\F_1\times \F_2$} 
\[
\begin{array}{@{}lcll}
&&<<u,v>, <e,f>> \in \overline{\F_1\times \F_2}(R) \\&\iff& <u,e> \in
\overline{\F_1}(R) \text{ and }  <v,f> \in  \overline{\F_2}(R)&\text{(def. $\overline{\F_1\times \F_2}$)}\\[1.5ex]
&\iff&  <u,e>  \in R_{\F_1 \lhd \G} \text{ and } <v,f>  \in  R_{\F_2 \lhd \G}& \text{(ind. hyp.)}\\[1.5ex]
&\iff& <u,e>  = <c , \delta_{\F_1\lhd \G}(\xi^{\F_1}_{c})> \text{ and }
<v,f>  = <c' , \delta_{\F_2\lhd \G}(\xi^{\F_2}_{c'})>&\text{(def. $R_{\F_i\lhd \G}$)}\\[1.5ex]
&\iff& <u,v> = <c,c'> \text{ and }  <e,f> = \delta_{\F_1\times
\F_2\lhd \G}(l(\xi^{\F_1}_{c}) \oplus r(\xi^{\F_2}_{c'}))&\text{(def. $\delta_{\F\lhd \G}$)}\\[1.5ex]
&\iff& <u,v> = <c,c'> \text{ and }  <e,f> = \delta_{\F_1\times
\F_2\lhd \G}(\xi^{\F_1\times \F_2}_{<c,c'>})&\text{(def. $\xi^\F$)}\\[1.5ex]
&\iff&  <<u,v>, <e,f>> \in R_{\F_1\times \F_2\lhd \G} 
\end{array}
\]

\fbox{$\F=\F_1\myplus \F_2$}, 
\fbox{$\F=\F_1^A$} and 
\fbox{$\F=\pow \F_1$}: similar to $\F_1\times \F_2$.

\item[\mycirc{$2$}]  We want to prove that $<g(s_i),
\delta_\G(\expr{s_i})> \in R_{\G\lhd \G}$. For that, we must show that
$g(s_i) \in \G<s>$ and $\delta_\G(\expr{s_i}) =
\delta_\G(\xi^\G_{g(s_i)})$. The former follows by definition of $<s>$,
whereas for the latter we observe that:
{\small
\begin{center}
\begin{longtable}{lll}
&\ \ \ \; $\delta_\G(\expr{s_i})$ \\
&$= \delta_\G((\mu x_i. \gamma^\G_{g(s_i)})  \{A^0_{1}/x_{1}\} \ldots
\{A^{n-1}_{n}/x_{n}\}) $&$ \text{(def. of $\expr{s_i}$)}$\\[1.5ex]
&$=\delta_\G(\mu x_i. \gamma^\G_{g(s_i)}  \{A^0_{1}/x_{1}\} \ldots
\{A^{i-2}_{i-1}/x_{i-1}\}  \{A^i_{i+1}/x_{i+1}\}
\ldots\{A^{n-1}_{n}/x_{n}\})$\\[1.5ex]
&$=\delta_\G(\gamma^\G_{g(s_i)}  \{A^0_{1}/x_{1}\} \ldots
\{A^{i-2}_{i-1}/x_{i-1}\}  \{A^i_{i+1}/x_{i+1}\}
\ldots\{A^{n-1}_{n}/x_{n}\}[A^n_i/x_i])$&$  \text{(def. of
$\delta_\G$)}$\\[1.5ex]
&$=\delta_\G(\gamma^\G_{g(s_i)}  \{A^0_{1}/x_{1}\} \ldots
\{A^{i-2}_{i-1}/x_{i-1}\}  \{A^i_{i+1}/x_{i+1}\}
\ldots\{A^{n-1}_{n}/x_{n}\}\{A^n_i/x_i\})$&$  \text{($[A^n_i/x_i] =
\{A^n_i/x_i\}$)}$\\[1.5ex]
&$=\delta_\G(\gamma^\G_{g(s_i)}  \{A^0_{1}/x_{1}\} \ldots
\{A^{i-2}_{i-1}/x_{i-1}\} \{A^n_i/x_i\} \{A^i_{i+1}/x_{i+1}\}
\ldots\{A^{n-1}_{n}/x_{n}\})$&\\[1.5ex]
&$=\delta_\G(\xi^\G_{g(s_i)})$
\end{longtable}
\end{center}
}%
Here, note that $[A^n_i/x_i] = \{A^n_i/x_i\}$, because $A^n_i$ has no free variables.
The last two steps follow, respectively, because $x_i$ is not free in $A^i_{i+1},\ldots, A^{n-1}_{n}$ and: 
\begin{eqnarray}\nonumber
 && \{A^n_i/x_i\} \{A^i_{i+1}/x_{i+1}\} \ldots\{A^{n-1}_{n}/x_{n}\} \\
 \nonumber  &=& \{A^{i-1}_i \{A^i_{i+1}/x_{i+1}\} \ldots\{A^{n-1}_{n}/x_{n}\} /x_i\} \{A^i_{i+1}/x_{i+1}\} \ldots\{A^{n-1}_{n}/x_{n}\} \\[1.5ex]
 \label{eq:subst1} &=& \{A^{i-1}_i/x_i\} \{A^i_{i+1}/x_{i+1}\} \ldots\{A^{n-1}_{n}/x_{n}\}
  \end{eqnarray}
Equation~(\ref{eq:subst1}) uses the syntactic identity
\begin{equation}\label{eq:subst2}A\{B\{C/y\}/x\}\{C/y\} =
A\{B/x\}\{C/y\},\ \ \ y \text{ not free in $C$}\end{equation} 
\end{itemize}
\qed
Let us illustrate the construction appearing in the proof of Theorem~\ref{kleene1} by some examples. 
These examples will illustrate the similarity with the proof of Kleene's Theorem presented in most textbooks, where a regular expression denoting the language recognized by a state of a deterministic automaton is built using a system of equations.

Consider the following
deterministic automaton over $A=\{a,b\}$, whose transition function $g$ is
given by the following picture ($\xymatrix{*++[o][F=]{s}}$
represents that the state $s$ is final):
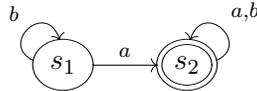
\[
\xymatrix{*++[o][F]{s_1} \ar[r]^{a}\ar@(l,u)[]^{b}&*++[o][F=]{s_2}
\ar@(r,u)[]_{a, b}}
\]
We define $A_1 = \mu x_1  . \gamma^\D_{g(s_1)}$ and
$A_2 = \mu x_2  . \,  \gamma^\D_{g(s_2)}$ where
\[
 \gamma^\D_{g(s_1)} = l<0> \oplus r<b(x_1) \oplus a(x_2)> \;\;\;\;\;
\gamma^\D_{g(s_2)} =
l<1> \oplus r<a(x_2) \oplus b(x_2)>
\]
We have $A^2_1 = A_1 \{A^1_2/x_2\}$ and $A^2_2 = A_2 \{A^0_1/x_1\} $.
Thus, $\expr {s_2} = A_2$ and, since $A^1_2=A_2$,  $\expr {s_1}$ is the expression 
\[
 \mu x_1 . \, l<0> \oplus r<b(x_1) \oplus a(\mu x_2 . \, l<1> \oplus r<a(x_2) \oplus
b(x_2)>
)>\]
By construction we have $s_1 \sim \expr{s_1}$ and
$s_2 \sim \expr{s_2}$.

For another example, take the following partial automaton, also over a two letter alphabet $A=\{a,b\}$:
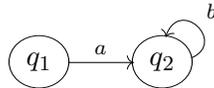
\[
\xymatrix{*++[o][F]{q_1} \ar[r]^{a}& *++[o][F]{q_2}
\ar@(r,u)[]_{b}}
\]
In the graphical representation of a partial
automaton $(S,p)$ we omit transitions for which $p(s)(a)
=\kappa_1(*)$. In this case, this happens in $q_1$ for the input
letter $b$ and in $q_2$ for $a$.

We will have the equations
\[
\begin{array}{l}
A_1 = A^0_1 = A^1_1 =  \mu x_1  . b(l[*]) \oplus a(r[x_2]) \\[1ex]
A_2 = A^0_2 = A^1_2 =  \mu x_2  . a(l[*]) \oplus b(r[x_2])
\end{array}
\]
Thus:
\[
\begin{array}{l}
\expr{s_1} = A^2_1 = \mu x_1 . \, b(l[*]) \oplus a(r[\mu x_2 . \,
a(l[*]) \oplus b(r[x_2])])\\
\expr{s_2} = \mu x_2  . a(l[*]) \oplus b(r[x_2])
\end{array}
\]
Again we have $s_1 \sim \expr{s_1}$ and $s_2 \sim \expr {s_2}$. 

As a last example, let us consider the following non\hyph deterministic automaton, over a one letter alphabet $A=\{a\}$:
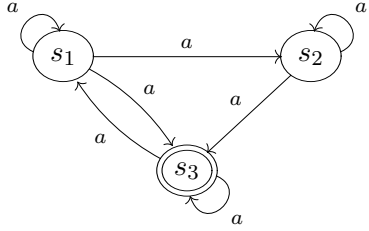
\[
\xymatrix{
*++[o][F]{s_1} \ar@(l,u)[]^{a}\ar@/^/[dr]^{a}\ar[rr]^{a}&& *++[o][F]{s_2}\ar@(r,u)[]_{a} \ar[dl]_{a} \\
&*++[o][F=]{s_3}\ar@(r,d)[]^{a}\ar@/^/[ul]^{a}&}
\]
We start with the equations:
\[
\begin{array}{lcl}
A_1&=&\mu x_1. l<0> \oplus r<a(\{x_1\} \oplus \{x_2\} \oplus \{x_3\})>\\
A_2&=&\mu x_2. l<0> \oplus r<a( \{x_2\} \oplus \{x_3\})>\\
A_3&=&\mu x_3. l<1> \oplus r<a(\{x_1\} \oplus \{x_3\})>
\end{array}
\]
Then we have the following iterations:
\[
\begin{array}{l}
A_1^1 = A_1\\[1ex]
A_1^2 = A_1 \{A_2^1/x_2\} = \mu x_1. l<0> \oplus r<a(\{x_1\} \oplus \{A_2\} \oplus \{x_3\})>\\[1ex]
A_1^3 = A_1 \{A_2^1/x_2\} \{A_3^2/x_3\} =  \mu x_1. l<0> \oplus r<a(\{x_1\} \oplus \{(A_2\{A_3^2/x_3\})\} \oplus \{A_3^2\})>\\
\\
A_2^1 =  A_2 \{A_1/x_1\} = A_2\\[1ex]
A_2^2 = A_2 \{A_1/x_1\} = A_2\\[1ex]
A_2^3 = A_2 \{A_1/x_1\}\{A_3^2/x_3\} = \mu x_2. l<0> \oplus r<a( \{x_2\} \oplus \{A_3^2\})>\\
\\
A_3^1 = A_3  \{A_1/x_1\} = \mu x_3. l<1> \oplus r<a(\{A_1\} \oplus \{x_3\})>\\[1ex]
A_3^2 = A_3  \{A_1/x_1\}\{A_2^1/x_2\} = \mu x_3. l<1> \oplus r<a(\{(A_1\{A_2^1/x_2\}) \} \oplus \{x_3\})>\\[1ex]
A_3^3 = A_3^2\\
\end{array}
\]
This yields the following expressions:
{\small
\[
\begin{array}{@{}l}
\expr{s_1} =  \mu x_1. l<0> \oplus r<a(\{x_1\} \oplus \{\expr{s_2}\} \oplus \{\expr{s_3}\})>\\
\expr{s_2} = \mu x_2. l<0> \oplus r<a( \{x_2\} \oplus \{\expr{s_3}\})>\\ 
\expr{s_3} = \mu x_3. l<1> \oplus r<a(\{ \mu x_1. l<0> \oplus r<a(\{x_1\} \oplus \{\mu x_2. l<0> \oplus r<a( \{x_2\} \oplus \{x_3\})>\} \oplus \{x_3\})> \} \oplus \{x_3\})>\\
\end{array}
\]
}
\subsection{Finite systems for expressions}\label{sec:synthesis}
Next, we prove the converse of
Theorem~\ref{kleene1}, that is, we
show how to construct a \emph{finite}
$\G$-coalgebra $(S,g)$ from an arbitrary expression $\E\in\Exp_\G$, such
that there exists a state $s\in S$ with $\E \sim_\G s$.

The immediate way of
obtaining a coalgebra from an expression $\E\in \Exp_\G$ is to compute
the subcoalgebra $<\E>$, since we have provided the set $\Exp_\G$ with a
coalgebra structure $\delta_\G\colon \Exp_\G\to \G(\Exp_\G)$. 
 However, the subcoalgebra
generated by an expression $\E\in \Exp_\G$ by repeatedly
applying
$\delta_\G$ is, in general, infinite. Take for instance the
deterministic expression $\E_1 = \mu x.\, r<a(x \oplus \mu y.\,
r<a(y)>)>$ (for simplicity, we consider $A=\{a\}$ and below we will
write, in the second component of $\delta_\D$, an expression $\E$
instead of the function mapping $a$ to $\E$) and
observe that:
\begin{center}
\scalebox{.975}{$
\begin{array}{@{}lcl}
\delta_\D(\E_1) &=& <0, \E_1 \oplus \mu y.\, r<a(y)>>\\
\delta_\D(\E_1 \oplus \mu y.\, r<a(y)>) &=& <0, \E_1 \oplus \mu
y.\, r<a(y)> \oplus \mu y.\, r<a(y)>>\\
\delta_\D(\E_1 \oplus \mu y.\, r<a(y)> \oplus \mu y.\, r<a(y)>) &=&
<0, \E_1 \oplus \mu
y.\, r<a(y)> \oplus \mu y.\, r<a(y)>\oplus \mu y.\, r<a(y)>>\\
&\vdots&
\end{array}$
}
\end{center}
As one would expect, all the new states are equivalent  and will be
identified by $\mathbf{beh}$ (the morphism into the final coalgebra). However, the 
function $\delta_\D$ does not make any state 
identification and thus yields an infinite coalgebra.

This phenomenon occurs also in classical regular expressions. It was
shown in~\cite{Brz64} that normalizing
the expressions using the axioms for associativity, commutativity and
idempotency was enough to guarantee finiteness\footnote{Actually, to
guarantee finiteness, similar to classical regular expressions, it
is enough to eliminate double occurrences of expressions $\E$ at the
outermost level of an expression $\cdots \oplus \E \oplus \cdots \oplus \E
\oplus \cdots$ (and to do this one needs the ACI axioms). Note that
this is weaker than taking expressions modulo the ACI
axioms: for instance, the expressions $\E_1\oplus \E_2$ and $\E_2\oplus
\E_1$, for $\E_1\neq \E_2$, would not be identified in the process
above. }. We will show in this
section that this also holds in our setting. 

Consider the following axioms (only the first three are essential, but we include
the fourth to obtain smaller coalgebras):
\[
\begin{array}{llll}
\mathit{(Associativity)} &  \E_1 \oplus (\E_2 \oplus \E_3)  \equiv
(\E_1 \oplus
\E_2) \oplus \E_3\\
\mathit{(Commutativity)} &  \E_1 \oplus \E_2  \equiv \E_2 \oplus
\E_1\\
\mathit{(Idempotency)}   &  \E \oplus \E  \equiv \E\\
\mathit{(Empty)}        &   \emp \oplus \E \equiv \E
\end{array}
\]
We define the relation $\equiv_{\mathit{ACIE}} \subseteq \Exp_{\F\lhd
\G}\times \Exp_{\F\lhd \G}$, written infix, as the least equivalence
relation containing the four identities above. The relation
$\equiv_{\mathit{ACIE}}$ gives rise to the (surjective) equivalence
map $[\E]_{\mathit{ACIE}} = \{\E'\mid \E\equiv_{\mathit{ACIE}}
\E'\}$. The following diagram shows the maps defined so far:
\[
\xymatrix{
\Exp_{\F\lhd \G}\ar[d]_{\delta_{\F\lhd \G}} \ar[rr]^-{[-]_{\mathit{ACIE}}}
&&
\Exp_{\F\lhd \G}/_{\equiv_{\mathit{ACIE}}}\\
\F(\Exp_{\G})\ar[rr]_-{\F([-]_{\mathit{ACIE}})} && \F (\Exp_\G/_{\equiv_{\mathit{ACIE}}}) 
}
\]
In order to complete the diagram, we next prove that $\equiv_{\mathit{ACIE}}$ is contained in the kernel
of $\F ([-]_{\mathit{ACIE}}) \circ \delta_{\F\lhd \G}$\footnote{This is
equivalent to prove that $\Exp_{\F\lhd \G}/_{\equiv_{\mathit{ACIE}}}$, together with
$[-]_{\mathit{ACIE}}$, is the coequalizer of the projection morphisms from $\equiv_{\mathit{ACIE}}$ to
$\Exp_{\F\lhd \G}$.}. 

This will
guarantee the existence of a function $$\overline{\delta}_{\F\lhd
\G}\colon
\Exp_{\F\lhd \G}/_{\equiv_{\mathit{ACIE}}} \to
\F(\Exp_\G/_{\equiv_{\mathit{ACIE}}})$$ which, when $\F=\G$, provides
$\Exp_\G/_\equiv$ with a coalgebraic structure $$\overline{\delta}_\G\colon
\Exp_\G/_{\equiv_{\mathit{ACIE}}} \to
\G(\Exp_\G/_{\equiv_{\mathit{ACIE}}})$$ (as before we write
$\overline{\delta}_\G$ for $\overline{\delta}_{\G\lhd \G}$) and which
makes $[-]_{\mathit{ACIE}}$ a homomorphism of coalgebras. 
\begin{lemma}\label{lemma:equivACIE}
Let $\G$ and $\F$ be non\hyph deterministic functors, with $\F\lhd \G$. For all $\E_1,\E_2\in \Exp_{\F\lhd \G}$,
\[
\E_1\equiv_{\mathit{ACIE}} \E_2 \Rightarrow (\F ([-]_{\mathit{ACIE}})) (\delta_{\F\lhd
\G}(\E_1)) = (\F ([-]_{\mathit{ACIE}})) (\delta_{\F\lhd
\G}(\E_2))
\]
\end{lemma}
\proof 
In order to improve readability, in this proof we will use $[-]$
to denote $[-]_{\mathit{ACIE}}$.  

It is enough to prove that for all $x_1, x_2, x_3 \in
\F(\Exp_{\G})$: 
\begin{itemize}
\item[\mycirc{1}] $\F([-])(\Plus_{\F\lhd \G}(\Plus_{\F\lhd
\G}(x_1, x_2),x_3)) = \F([-])(\Plus_{\F\lhd
\G}(x_1,\Plus_{\F\lhd \G}(x_2,x_3)))$
\item[\mycirc{2}] $\F([-])(\Plus_{\F\lhd
\G}(x_1,x_2))= \F([-])(\Plus_{\F\lhd \G}(x_2,x_1))$
\item[\mycirc{3}] $\F([-])(\Plus_{\F\lhd
\G}(x_1,x_1))=\F([-])(x_1)$
\item[\mycirc{4}] $\F([-])(\Plus_{\F\lhd \G}( \Empty_{\F\lhd
\G},x_1)) =\F([-])(x_1)$
\end{itemize}
By induction on the structure of $\F$. We illustrate a few cases, the
omitted ones are proved in a similar way.

\noindent\fbox{$\F=\id$}  $x_1, x_2, x_3\in \Exp_{\G}$
 
$$\begin{array}{@{}llll}
\mycirc{1}&&[\Plus_{\id\lhd \G}(\Plus_{\id\lhd
\G}(x_1,x_2),x_3)]\\ &=& [(x_1 \oplus x_2) \oplus x_3] & \text{(def.
$\Plus$)}\\
&=& [x_1 \oplus (x_2 \oplus x_3)] &\text{($\mathit{Associativity}$)}\\
  &=& [\Plus_{\id\lhd
\G}(x_1,\Plus_{\id\lhd
\G}(x_2,x_3))]  &\text{(def.
$\Plus$)}\\[1.5ex]
\mycirc{4}&& [\Plus_{\id\lhd \G}( \Empty_{\id\lhd
\G},x_1) ] \\
&=&[\emp \oplus x_1] &\text{(def.
$\Plus$ and $\Empty$)}\\
 &=& [x_1] &\text{($\Empty$)}
\end{array}$$

\noindent\fbox{$\F=\F_1\times \F_2$}  $x_1 = <u_1,v_1>, x_2 =
<u_2,v_2> \in (\F_1\times
\F_2)(\Exp_{\G})$
$$\begin{array}{@{}llll}
\mycirc{2}&&(\F_1\times \F_2) ([-])(\Plus_{\F_1\times \F_2\lhd
\G}(<u_1,v_1>,<u_2,v_2>))\\
&= & <\F_1 ([-])(\Plus_{\F_1\lhd \G}(u_1,u_2)),\F_2 ([-])(\Plus_{\F_2\lhd
\G}(v_1,v_2))>&\text{(def. $\Plus$)}\\
&=&<\F_1 ([-])(\Plus_{\F_1\lhd \G}(u_2,u_1)),\F_2 ([-])(\Plus_{\F_2\lhd
\G}(v_2,v_1))>&\text{(ind. hyp.)}\\
&=& (\F_1\times \F_2)([-])(\Plus_{\F_1\times \F_2\lhd
\G}(<u_2,v_2>,<u_1,v_1>)) &\text{(def. $\Plus$)}\\[1.5ex]
\mycirc{3}&&(\F_1\times \F_2)([-])(\Plus_{\F_1\times \F_2\lhd
\G}(<u_1,v_1>,<u_1,v_1>)) \\
&= & <\F_1 ([-])(\Plus_{\F_1\lhd \G}(u_1,u_1)),\F_2 ([-])(\Plus_{\F_2\lhd
\G}(v_1,v_1))>&\text{(def. $\Plus$)}\\
&=&<\F_1 ([-])(u_1), \F_2 ([-])(v_1)>&\text{(ind. hyp.)}\\
&=& ({\F_1\times \F_2})([-])(<u_1,v_1>) & 
\end{array}
$$

\noindent\fbox{$\F=\pow \F_1$}  $x_1, x_2, x_3 \in \pow \F_1(\Exp_{\G})$
$$\begin{array}{p{.5cm}cll}
\mycirc{1}&&\pow \F_1 ([-])(\Plus_{\pow \F_1\lhd \G}(x_1,\Plus_{\pow \F_1\lhd
\G}(x_2,x_3)))\\
&=& \pow \F_1 ([-])(x_1\cup(x_2\cup x_3))&\text{(def. $\Plus$)}\\
&=& \pow \F_1 ([-])((x_1\cup x_2)\cup x_3)&\\
&=& \pow \F_1 ([-])(\Plus_{\pow \F_1\lhd \G}(\Plus_{\pow \F_1\lhd
\G}(x_1,x_2),x_3))&\text{(def. $\Plus$)}\\
\end{array}
$$
In the last but one step, we use the fact
that, for any set $X$, $(\pow(X), \cup, \emptyset)$ is a join\hyph semilattice
(hence, $x_1\cup(x_2\cup x_3)=(x_1\cup x_2)\cup x_3$). Due to this
fact, in the case $\F=\pow \F_1$, in this particular proof, the
induction hypothesis will not be
used.
\qed
Thus, we have a well-defined function $$\overline{\delta}_{\F\lhd
\G}\colon\Exp_{\F\lhd \G}/_{\equiv_{\mathit{ACIE}}}\to \F(
\Exp_\G/_{\equiv_{\mathit{ACIE}}})$$ such that $\overline{\delta}_{\F\lhd
\G} ([\E]_{\mathit{ACIE}}) = (\F [-]_{\mathit{ACIE}})({\delta}_{\F\lhd
\G} (\E))$. 

We are ready to state and prove the second half of Kleene's theorem. 

\begin{theorem}\label{kleene2}
Let $\G$ be a non\hyph deterministic functor. For every $\E\in \Exp_\G$,
there exists $\Delta_\G(\E)=(S,g)$ such that $S$ is finite and there exists $s\in
S$ with $\E\sim s$.
\end{theorem}
\proof 
For every $\E\in  \Exp_\G$, we set $\Delta_\G(\E) =
<[\E]_{\mathit{ACIE}}>$ (recall that $<s>$ denotes the smallest subcoalgebra
generated by $s$). First note that, by
Lemma~\ref{lemma:equivACIE}, the map $[-]_\mathit{ACIE}$ is a
homomorphism and thus $\E\sim [\E]_\mathit{ACIE}$. We prove, for every
$\E\in  \Exp_\G$, that the subcoalgebra $<[\E]_{\mathit{ACIE}}> =
(V, \overline\delta_{\G})$  has a finite state space $V$ (here, $ \overline\delta_\G$
actually stands for the restriction of $ \overline\delta_\G$ to $V$).
Again, in order to improve readability, below we will use $[-]$ to denote $[-]_{\mathit{ACIE}}$.  

More precisely, we prove, for all $\E \in  \Exp_{\F\lhd \G}$, the
following inclusion 
\begin{equation}\label{eq:incl_kleene2}
V \subseteq \overline V = \{[\E_1\oplus \ldots \oplus \E_k] \mid
\E_1,\ldots,\E_k \in \mathit{cl}(\E) \text{ all distinct}, \E_1,\ldots, \E_k
\in \Exp_\G\}
\end{equation}
Here,  if $k=0$ we take the sum above to be $\emp$ and $\mathit{cl}(\E)$ denotes the smallest set containing all subformulas of $\E$ and the unfoldings of $\mu$ (sub)formulas, that is, the smallest subset satisfying:
\[
\begin{array}{lcllcl}
\mathit{cl}(\emp) &= &\{\emp\}&\mathit{cl}(l[\E_1]) &=& \{l[\E_1]\}
\cup \mathit{cl}(\E_1)\\
\mathit{cl}(\E_1\oplus \E_2) &=&  \{\E_1\oplus\E_2\}\cup
\mathit{cl}(\E_1)\cup \mathit{cl}(\E_2)&\mathit{cl}(r[\E_1]) &=&
\{r[\E_1]\} \cup \mathit{cl}(\E_1)\\
\mathit{cl}(\mu x. \E_1) &=& \{\mu x.\,\E_1\}\cup \mathit{cl}(\E_1
[\mu x. \E_1/x]) & \mathit{cl}(a(\E_1)) &=& \{a(\E_1)\} \cup
\mathit{cl}(\E_1)\\
\mathit{cl}(l<\E_1>) &=& \{l<\E_1>\} \cup \mathit{cl}(\E_1)&\mathit{cl}(\{\E_1\}) &=& \{\{\E_1\}\} \cup
\mathit{cl}(\E_1)
\\
\mathit{cl}(r<\E_1>) &=& \{r<\E_1>\} \cup
\mathit{cl}(\E_1)\end{array}
\]
Note that the set $\mathit{cl}(\E)$ is finite (the number of different unfoldings is finite) and has the property $\E\in
\mathit{cl}(\E)$.

We prove the inclusion in equation~(\ref{eq:incl_kleene2}) in the following way. First,
we observe that $[\E] \in \overline V$, because $\E \in \mathit{cl}(\E)$. Then,
we prove that $(\overline V, \overline\delta_\G)$ (again, $\overline\delta_\G$
actually stands for the restriction of $\overline\delta_\G$ to $\overline V$) is a subcoalgebra of
$(\Exp_\G,\overline\delta_\G)$. Thus, $V \subseteq \overline V$, since $V$, the
state space of $<[\E]>$ is equal to the intersection of all
subcoalgebras of $(\Exp_\G,\overline\delta_\G)$ containing $[\E]$.
 
To prove that $(\overline V, \overline\delta_\G)$ is a subcoalgebra we prove
that, for $\E_1,\ldots, \E_k
\in \Exp_{\F\lhd \G}$,
\begin{equation}\label{eq:kleene2_1_5}
\E_1,\ldots,\E_k \in \mathit{cl}(\E) \text{ all distinct} \Rightarrow \overline\delta_{\F\lhd
\G}([\E_1\oplus\ldots \oplus \E_k]) \in \F(\overline V)
\end{equation}
The intended
result then follows by taking $\F=\G$.

We first prove two auxiliary results, by induction on the structure of
$\F$:
\[
\begin{array}{@{}cl}
\mycirc{1}& (\F [-])(\Empty_{\F\lhd \G}) \in \F(\overline
V)\\[1.2ex]
\mycirc{2}& (\F [-])(\Plus_{\F\lhd \G}(u,v)) \in \F(\overline V)
\Leftrightarrow (\F [-])(u) \in \F(\overline V)\text{ and } (\F [-])(v) \in
\F(\overline V)
\end{array}
\]
for $u,v\in \F(\Exp_\G)$.

\noindent\fbox{$\F=\id$} 
\[
\begin{array}{clr}
\mycirc{1}& (\F [-])(\Empty_{\F\lhd \G}) = [\emp] \in \overline V\\[1.2ex]
\mycirc{2}& (\F [-])(\Plus_{\F\lhd \G}(u,v)) = [u \oplus v]\in \overline
V  
\Leftrightarrow [u] \in \overline V\text{ and } [v] \in
\overline V& u,v\in \Exp_\G
\end{array}
\]
The right to left implication follows because, using
 the $(\mathit{Associativity})$, $(\mathit{Commutativity})$ and
$(\mathit{Idempotency})$ axioms, we can rewrite $u\oplus v$ as
$\E_1\oplus\ldots\oplus \E_k$, with all $\E_1,\ldots,\E_k \in \mathit{cl}(\E)$ distinct. 

\noindent\fbox{$\F=\B$} 
\[
\begin{array}{@{}cl@{\ }r}
\mycirc{1}& (\B [-])(\Empty_{\B\lhd \G}) = \bot_\B  \in \B(\overline
V)\\[1.2ex]
\mycirc{2}& (\B [-])(\Plus_{\B\lhd \G}(u,v)) = u \vee v \in \B(\overline
V)
\Leftrightarrow u \in \B(\overline V)\text{ and } v \in
\B(\overline V)& u,v\in \B(\Exp_\G)= \B
\end{array}
\]
\fbox{$\F=\F_1\times \F_2$} 
\[
\begin{array}{@{}clr}
\mycirc{1}& (\F_1\times \F_2 [-])(\Empty_{\F_1\times \F_2\lhd \G}) \\&=
<(\F_1 [-])(\Empty_{\F_1\lhd \G}),(\F_2 [-])(\Empty_{\F_2\lhd \G})> \in \F_1\times \F_2(\overline
V)\\[1.2ex]
\mycirc{2}& (\F_1\times \F_2 [-])(\Plus_{\F_1\times \F_2\lhd
\G}(<u_1,u_2>,<v_1,v_2>))= 
\\
&<(\F_1 [-])(\Plus_{\F_1 \lhd
\G}(u_1,v_1)), (\F_2 [-])(\Plus_{\F_2\lhd
\G}(u_2,v_2))>  \in \F_1\times \F_2(\overline
V)\\ &
\stackrel{\mathit{(IH)}}\Leftrightarrow u_1,v_1 \in \F_1(\overline V)
\text{ and } u_2,v_2 \in \F_2(\overline V) \\ & 
\Leftrightarrow <u,v> \in
\F_1\times \F_2(\overline V),\ \ \ u=<u_1,u_2>,v=<v_1,v_2>\in \F_1\times \F_2(\Exp_\G)
\end{array}
\]
\fbox{$\F=\F_1\myplus \F_2$} and 
\fbox{$\F=\F_1^A$}: similar to $\F_1\times \F_2$.

\ \\
\fbox{$\F=\pow \F_1$} 
\[
\begin{array}{@{}clr}
\mycirc{1}& (\pow \F [-])(\Empty_{\pow \F\lhd \G}) = \emptyset
 \in \pow \F(\overline
V)\\[1.2ex]
\mycirc{2}& (\pow \F [-])(\Plus_{\pow \F\lhd
\G}(u,v))= 
((\pow \F [-]) (u) \cup (\pow \F [-])(v))\in \pow \F(\overline V) \\ &
\Leftrightarrow (\pow \F [-] (u))\in \pow \F(\overline V)
\text{ and } (\pow \F [-] (v))\in \pow \F(\overline V) \end{array}
\]
Using $\mycirc{2}$, we can simplify our proof goal
(equation~(\ref{eq:kleene2_1_5})) as follows:
\[
\overline\delta_{\F\lhd
\G}([\E_1\oplus\ldots \oplus \E_k]) \in \F(\overline V)
\Leftrightarrow (\F[-])(\delta_{\F\lhd
\G}(\E_i)) \in \F(\overline V), \ \E_i\in \mathit{cl}(\E),\ i=1,\ldots,k
\]
Next, using induction on the product of types of expressions and expressions
(using the order in equation~(\ref{ind:order5})), $\mycirc 1$
and $\mycirc 2$, we prove that $(\F[-])(\delta_{\F\lhd
\G}(\E_i)) \in \F(\overline V)$, for any $\E_i\in \mathit{cl}(\E)$. 
\begin{longtable}{@{}ll}
$(\F[-])(\delta_{\F\lhd \G}  (\emp)) = (\F[-])(\Empty_{\F\lhd
\G})\in \F(\overline V)$ &$(\text{by }\mycirc 1)$\\[1.1ex]
$(\F[-])(\overline\delta_{\F\lhd \G}  (\E_1\oplus\E_2)) 
= (\F[-])(\Plus_{\F\lhd \G}(\delta_{\F\lhd
\G}  (\E_1), \delta_{\F\lhd \G}  (\E_2))\in \F(\overline V)$
&$(\text{$\mathit{IH}$ and }\mycirc 2)$\\[1.1ex]
$(\G[-])(\delta_{\G\lhd \G}  (\mu x. \E)) = (\G[-])(\delta_{\G\lhd \G} (\E[\mu x.
\E/x]))\in \G(\overline V)$&$(\text{$\mathit{IH}$})$\\[1.1ex]
$(\id [-])(\delta_{\id \lhd \G}(\E_i)) = [\E_i]\in \id(\overline V) \ \ \textrm{for} \ \G\neq
\id$ &$(\E_i\in \mathit{cl}(\E))$\\[1.1ex]
$(\B[-])\delta_{\B \lhd \G} (b) = b\in \B(\overline V)$&$(\B(\overline
V)=\B)$\\[1.1ex]
$(\F_1\times \F_2[-])(\delta_{\F_1\times \F_2 \lhd \G} (l<\E>))$ \\
$ \ = <
(\F_1[-])(\delta_{\F_1\lhd \G}(\E)), (\F_2 [-])(\Empty_{\F_2\lhd
\G}) >\in
\F_1\times \F_2(\overline V)$
&$(\text{$\mathit{IH}$ and }\mycirc 1)$\\[1.3ex]
$(\F_1\times \F_2[-])(\delta_{\F_1\times \F_2 \lhd \G} (r<\E>))$\\ $\ = <(\F_1
[-])(\Empty_{\F_1\lhd \G}),
(\F_2[-])(\delta_{\F_2\lhd \G}(\E)) >\in
\F_1\times \F_2(\overline V)
$&$(\text{$\mathit{IH}$ and }\mycirc 1)$\\[1.1ex]
$(\F_1\myplus \F_2[-])(\delta_{\F_1\myplus \F_2 \lhd \G} (l[\E])) =
\kappa_1((\F_1[-])(\delta_{\F_1
\lhd \G}(\E)))\in
\F_1\myplus \F_2(\overline V)$
&$(\text{$\mathit{IH}$})$\\[1.1ex]
$(\F_1\myplus \F_2[-])(\delta_{\F_1\myplus \F_2 \lhd \G} (r[\E])) =
\kappa_2((\F_2[-])(\delta_{\F_2
\lhd \G}(\E)))\in
\F_1\myplus \F_2(\overline V)
$&$(\text{$\mathit{IH}$})$\\[1.1ex]
$(\F^A[-])(\delta_{\F^A \lhd \G} (a(\E))) = \left(\lambda a' .\left\{
\begin{array}{ll}(\F[-])(\delta_{\F \lhd \G} (\E))
&\text{if }a=a'\\\Empty_{\F\lhd
\G}&\text{otherwise}\end{array}\right.\right)\in
\F^A(\overline V)
$&$(\text{$\mathit{IH}$ and }\mycirc
1)$\\[2.1ex]
$(\pow \F[-])(\delta_{\pow \F\lhd \G} (\{\E\})) = \{\,(\F[-])(\delta_{\F
\lhd \G}(\E))\,\}\in
\pow \F(\overline V)$
&$(\text{$\mathit{IH}$})$
\end{longtable}
\qedhere

\subsubsection{Examples}\label{sec:examples_synt}
n this subsection we will illustrate the construction described in the proof of Theorem~\ref{kleene2}: given an expression $\E\in \Exp_\G$ we construct a $\G$-coalgebra $(S,g)$ such that there is $s\in S$ with $s\sim \E$. For simplicity, we will consider
deterministic and partial automata expressions over $A=\{a,b\}$.

Let us start by showing the synthesised automata for the most simple
deterministic expressions -- $\emp$, $l<0>$ and $l<1>$.
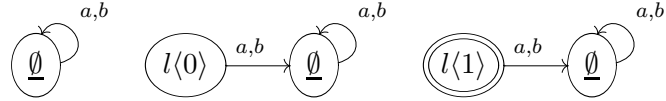
\[
\xymatrix{*++[o][F]{\emp} \ar@(r,u)[]_{a,b}} \ \ \
\xymatrix{*++[o][F]{l<0>} \ar[r]^{a,b}&  *++[o][F]{\emp } \ar@(r,u)[]_{a,b}}
\ \ \ \xymatrix{*++[o][F=] {l<1>} \ar[r]^{a,b}&  *++[o][F]{\emp}
\ar@(r,u)[]_{a,b}}
\]
The first two automata recognize the
empty language $\emptyset$ and the last the language $\{\epsilon \}$ containing only the empty
word.

We note that the 
generated automata are not minimal (for instance, the automata for  $l<0>$ and
$\emp$ are bisimilar). Our goal has been to generate a finite
automaton from an expression. From this the minimal automaton
can always be obtained by identifying bisimilar states.

The following automaton, generated from the expression
$r<a(l<1>)>$, recognizes the language \{a\},
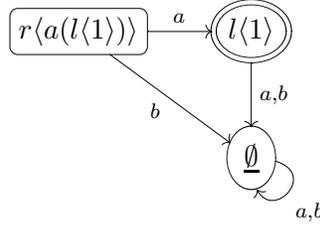
\[
\xymatrix{ *+[F-:<3pt>]{r<a(l<1>)>} \ar[r]^-{a}\ar[dr]_{b}& *++[o][F=]{l<1>}
\ar[d]^{a,b}  \\
& *++[o][F]{\emp}  \ar@(r,d)[]^{a,b}}
\]
For an example of an expression containing fixed points, consider $\E =
\mu x.\, r<a(l<0>\oplus l<1>\oplus x)>$. One can easily compute the synthesised
automaton:
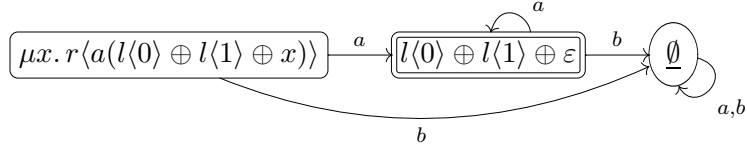
\[
\xymatrix{*+[F-:<3pt>]{\mu x.\, r<a(l<0>\oplus l<1>\oplus x)>}
\ar[r]^-{a}\ar@/{}_{2 pc}/[rr]_{b}&
*+[F=:<3pt>]{{l<0>\oplus l<1>\oplus \E}} \ar@(r,u)[]_{a}\ar[r]^-{b} &
*++[o][F]{\emp}  \ar@(r,d)[]^{a,b}}
\]
and observe that it recognizes the language $aa^*$. Here, the role of
the join\hyph semilattice structure is also visible: $l<0>\oplus l<1>\oplus
\E$ specifies that this state is supposed to be non-final ($l<0>$) and final ($l<1>$). The conflict of these two specifications  is solved, when they are combined with $\oplus$, using the semilattice: because $1\vee 0 = 1$ the state is set to be final.

As a last example of deterministic expressions consider
\scalebox{.95}{$\E_1 = \mu
x.\, r<a(x \oplus \mu y. r<a(y)>)>$}.
Applying $\delta_\D$ to $\E_1$ one gets the following (partial)
automaton:
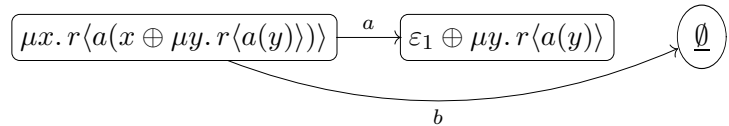
\[
\xymatrix{*+[F-:<3pt>]{\mu
x.\, r<a(x \oplus \mu y.\, r<a(y)>)>}
\ar[r]^-{a}\ar@/{}_{2 pc}/[rr]_{b}&
*+[F-:<3pt>]{\E_1 \oplus \mu y.\, r<a(y)>}
&*++[o][F]{\emp } }
\]
Calculating $\delta_\D(\E_1 \oplus \mu y.\, r<a(y)>)$ yields
\[
\begin{array}{lcl}
\delta_\D(\E_1 \oplus \mu y.\, r<a(y)>)& =& <0,t>\\
&\text{where}& t(a) = \E_1 \oplus \mu y.\,
r<a(y)>\oplus \mu y.\, r<a(y)>>\\
&& t(b) = \emp
\end{array}
\]
Note that the expression $\E_1 \oplus \mu y.\,
r<a(y)>\oplus \mu y.\, r<a(y)>$ is in the same equivalence class as $\E_1 \oplus \mu y.\, r<a(y)>$, which is a state that already exists.
As we saw in the beginning of Section~\ref{sec:expressions_coalg}, by only
applying $\delta_\D$, without $\mathit{ACI}$, one would always generate syntactically different states which instead of
the automaton computed now:
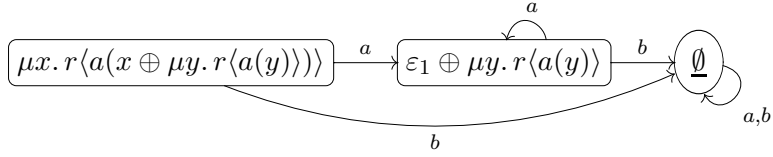
\[
\xymatrix{*+[F-:<3pt>]{\mu
x.\,r< a(x \oplus \mu y.\, r<a(y)>)>}
\ar[r]^-{a}\ar@/{}_{2 pc}/[rr]_{b}& *+[F-:<3pt>]{
\E_1 \oplus \mu y.\, r<a(y)>}\ar@(r,u)[]_-{a}\ar[r]^-{b}
&*++[o][F]{\emp} \ar@(r,d)[]^{a,b} }
\]
would yield the following infinite automaton (with $\E_2 = \mu y.\,
r<a(y)>$):
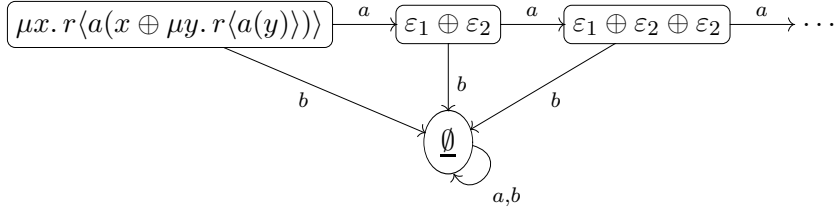
\[
\xymatrix{*+[F-:<3pt>]{\mu
x.\, r<a(x \oplus \mu y.\, r<a(y)>)>} \ar[r]^-{a}\ar[dr]_{b}&
*+[F-:<3pt>]{\E_1 \oplus \E_2}\ar[r]^-{a}\ar[d]^{b}&
*+[F-:<3pt>]{\E_1 \oplus \E_2 \oplus \E_2}\ar[r]^-{a}\ar[dl]^{b}&\ldots\\
&*++[o][F]{\emp} \ar@(r,d)[]^{a,b} &}
\]
Let us next see a few examples of synthesis for partial automata
expressions, where we will illustrate the role of $\bot$ and $\top$.
In the graphical representation of a partial
automaton $(S,p)$, we will omit transitions for inputs $a$ with
$g(s)(a) = \kappa_1(*)$ and we will draw $\xymatrix{*++[o][F]{s} \ar@{~>}[r]^-{a}&
*++[F]{g(s)(a)}}$ whenever $g(s)(a)\in\{\bot,\top\}$. Note however
that
$\bot\not\in S$ and $\top\not\in S$ and thus will have no defined
transitions. As before, let us first present the corresponding
automata for simple expressions -- $\emp$, $a(l[*])$,
$a(\emp)$  and $a(l[*]) \oplus b(l[*])$.
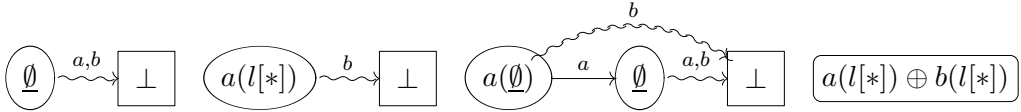
\[
\xymatrix{*++[o][F]{\emp} \ar@{~>}[r]^-{a,b}& *++[F]{\bot}} \ \ \
\xymatrix{*++[o][F]{a(l[*])} \ar@{~>}[r]^-{b}&  *++[F]{\bot }}
\ \ \ \xymatrix{*++[o][F] {a(\emp)}
\ar[r]^-{a}\ar@{~>}@/^0.7cm/[rr]^-{b}&  *++[o][F]{\emp} \ar@{~>}[r]^-{a,b}&
*++[F]{\bot}}\ \ \ \xymatrix{*+[F-:<3pt>]{a(l[*]) \oplus
b(l[*])} }
\]
Note how $\bot$ is used to encode
underspecification, working as a kind of deadlock state. In the first 
three expressions the behaviour for one or both of the
inputs is missing, whereas in the last expression the specification
is complete. 

The element $\top$ is used to deal with inconsistent
specifications. For instance, consider the expression $a(l[*])
\oplus b(l[*]) \oplus a(r[a(l[*]) \oplus b(l[*])])$. All inputs
are specified, but note that at the outermost level input $a$
appears in two different sub-expressions -- $a(l[*])$ and
$a(r[a(l[*]) \oplus b(l[*])])$ -- specifying at the same time that
input $a$ leads to successful termination and that it leads to a
state where $a(l[*]) \oplus b (l[*])$ holds, which is
contradictory, giving rise to the following automaton.
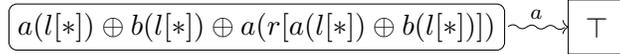
\[
\xymatrix{*+[F-:<3pt>]{a(l[*]) \oplus b(l[*]) \oplus a(r[a(l[*])
\oplus b(l[*])])
} \ar@{~>}[r]^-{a}& *++[F]{\top}}
\]

\section{A sound and complete axiomatization}\label{sec:axiom}

In the previous section, we have shown how to derive from the type of
a system, given by a functor $\G$, a language $\Exp_\G$ that allows for
specification of $\G$-behaviours. Analogously to Kleene's theorem, we
have proved the correspondence between the behaviours denoted by $\Exp_\G$
and locally finite $\G$-coalgebras. In this section, we will show how
to provide $\Exp_\G$ with a sound and complete axiomatization. Again,
the functor $\G$ will serve as a main guide for the definition. The
defined axiomatization is closely related to Kleene algebra (the set
of expressions has a join semilattice structure) and to the
axiomatization provided by Milner for CCS (uniqueness of fixed points
will be required). When instantiating the definition below to concrete
functors one will recover known axiomatizations, such as the one for
CCS mentioned above or the one for labelled transition systems (with
explicit termination) presented in~\cite{AcetoH92}. The latter will be
discussed in detail in Section~\ref{sec:appl}.

Next, we introduce an equational system  for expressions
of type $\F \lhd \G$. We define the relation
$\mathord{\equiv}\subseteq \Exp_{\F\lhd \G}\times \Exp_{\F\lhd \G}$, written
infix, as the least equivalence relation containing the
following identities:
\begin{enumerate}
\item $(\Exp_{\F \lhd \G}, \oplus, \emp)$ is a join-semilattice.
\begin{align*}
   &\E \oplus \E  \equiv \E && \mathit{(Idempotency)}\\
 &\E_1 \oplus \E_2  \equiv \E_2 \oplus \E_1 &&\mathit{(Commutativity)} \\
  &\E_1 \oplus (\E_2 \oplus \E_3)   \equiv (\E_1 \oplus
\E_2) \oplus \E_3 && \mathit{(Associativity)} \\
 &\emp \oplus \E \equiv \E &&\mathit{(Empty)}     
\end{align*}
\item $\mu$ is the unique fixed point.
\begin{align*}
&\gamma[\mu x.\gamma/x] \equiv \mu x . \gamma && \mathit{(FP)}  \\
&{\gamma[\E/x]  \equiv  \E \Rightarrow \mu
x.\gamma \equiv \E} &&\mathit{(Unique)}
\end{align*}
\item The join-semilattice structure propagates through the
expressions.
\[
\begin{array}{lcll@{\hspace{.5cm}}lcll}
\emp &\equiv& \bot_\B & (\B - \emp) &  b_1 \oplus b_2 &\equiv& b_1 \vee_\B b_2
&(\B - \oplus)  \\
  l<\emp> &\equiv& \emp & (\times - \emp - L) &   l<\E_1 \oplus \E_2>
&\equiv& l<\E_1> \oplus l<\E_2> &(\times - \oplus - L)\\
 r<\emp> &\equiv& \emp &(\times - \emp - R)  &
r<\E_1 \oplus \E_2> &\equiv& r<\E_1> \oplus r<\E_2> &(\times - \oplus - R) 
\\
  a(\emp) &\equiv& \emp &(-^A - \emp)  &
  a(\E_1 \oplus \E_2) &\equiv& a(\E_1) \oplus a(\E_2)&(-^A - \oplus)\\
&& &&  l[\E_1 \oplus \E_2] &\equiv& l[\E_1] \oplus l[\E_2] &(+ - \oplus - L) 
\\ &&&&     r[\E_1 \oplus \E_2] &\equiv& r[\E_1] \oplus
r[\E_2]&(+ - \oplus - R) \\
& &&&  l[\E_1] \oplus r[\E_2] &\equiv& l[\emp] \oplus r[\emp]
&{(+ - \oplus - \top)} 
\end{array}
\]
\item $\equiv$ is a congruence.
\begin{align*}
\E_1\equiv \E_2 \Rightarrow \E[\E_1/x]\equiv \E[\E_2/x]& \text{\ \ \ \
for $x$ free in $\E$}&(\mathit{Cong})
\end{align*}

\item $\alpha$-equivalence
\begin{align*}
\mu x.\gamma \equiv \mu y. \gamma[y/x] &\ \ \   \text{ for 
$y$  not free in $\gamma$}&& \mathit {(\alpha-\mathit{equiv})} 
\end{align*}
\end{enumerate}
It is important to remark that in the third group of rules
there does 
not exist any rule applicable to expressions of type $\pow \F$.

\begin{example}
Consider the non\hyph deterministic automata over the alphabet~$A=\{a\}$:
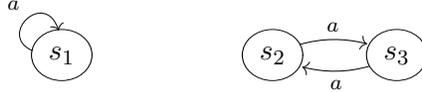
\[
\xymatrix{*++[o][F]{s_1}
\ar@(l,u)^-{a}}\hspace{2cm}\xymatrix{*++[o][F]{s_2} \ar@/^/[r]^-{a}&
*++[o][F]{s_3} \ar@/^/[l]^-{a}}
\] 
Applying $\expr -$ (as defined in the proof of Theorem~\ref{kleene1}) one can easily compute 
the expressions corresponding to $s_1$ and $s_2$:
\[
\begin{array}{l}
\E_1 = \expr{s_1} = \mu x_1. l<0> \oplus r<a(\{x_1\})>\\ 
\E_2 = \expr{s_2} = \mu y_1. l<0> \oplus r< a(\{\mu y_2. l<0> \oplus r<a(\{\mu y_1. l<0> \oplus r<a(\{y_2\})> \})>\})>
\end{array}
\]
We prove that $\E_2\equiv \E_1$. 
 In the following calculations let $\E = \mu x_1. r<a(\{x_1\})>$.
\[
\begin{array}{@{}lcll}
&& \E_2\equiv \E_1 \\
&\Leftrightarrow & r<a(\{\mu y_2. r<a(\{r<a(\{y_2\})>\})>\})>
\equiv \E & \text{($(\B-\emp)$, $(\times - \emp - L)$, $\mathit{(FP)}$ and $\mathit{(Empty)}$)}\\
&\Leftrightarrow& \mu y_2. r<a(\{r<a(\{y_2\})>\})> \equiv \E &\text{($(\mathit{FP})$ on $\E$ and $(\mathit{Cong})$ twice)}\\
&\Leftarrow &  r<a(\{r<a(\{\E\})>\})>  \equiv \E & \text{(uniqueness of fixed points)}
\\
&\Leftrightarrow &  r<a(\{\E\})> \equiv \E &\text{(fixed point axiom)}
\\
&\Leftrightarrow & \E \equiv \E
&\text{(fixed point axiom)}\end{array}
\]
Note that the $(\mathit{Cong})$ rule was used in almost every step.

For another example, consider the non\hyph deterministic automaton over the alphabet $A=\{a,b\}$:
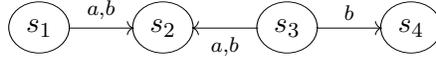
\[
\xymatrix{*++[o][F]{s_1} \ar[r]^-{a,b}& *++[o][F]{s_2} &
*++[o][F]{s_3}\ar[l]^-{a,b} \ar[r]^-{b}& *++[o][F]{s_4} }
\] 
Using the definition of $\expr -$ one can compute 
the following expressions for $s_1$, $s_2$, $s_3$ and $s_4$:
\[
\begin{array}{l}
\E_1 =  \expr{s_1} =\mu x_1. l<0> \oplus r<a(\{\E_2\})\oplus
b(\{\E_2\})> \\
\E_2 =\expr{s_2} = \mu x_2. l<0> \oplus \emp\\
\E_3 = \expr{s_3} =\mu x_3. l<0> \oplus r< a(\{\E_2\}) \oplus
b(\{\E_2\}\oplus\{\E_4\})>\\
\E_4 = \expr{s_4} =\mu x_4. l<0> \oplus \emp\\
\end{array}
\]
For $\E_2$ we calculate:
\[
\begin{array}{lcll}
\E_2 &\equiv& l<0>\oplus \emp& \text{($\mathit{FP}$)}\\
&\equiv& l<\emp> &  \text{($\mathit{Empty}$) and ($\B -
\emp$)}\\
&\equiv& \emp &  \text{($\times -
\emp - L$)}
\end{array}
\]
Similarly, one has that $\E_4\equiv \emp$.
Next, we prove $\E_1\equiv\E_3$:
\[
\hspace{-.3cm}\begin{array}{@{}lcll}
&& \E_1\equiv \E_3 \\
&\Leftrightarrow &l<0> \oplus r<a(\{\E_2\})\oplus b(\{\E_2\})>\equiv  l<0> \oplus r< a(\{\E_2\}) \oplus
b(\{\E_2\}\oplus\{\E_4\})> & \text{($\mathit{FP}$)}\\
&\Leftrightarrow& l<0> \oplus r<a(\{\emp\})\oplus
b(\{\emp\})>\equiv  l<0> \oplus r< a(\{\emp\}) \oplus
b(\{\emp\}\oplus\{\emp\})> &\text{($\E_2\equiv
\emp\equiv \E_4$)}\\
&\Leftrightarrow& l<0> \oplus r<a(\{\emp\})\oplus
b(\{\emp\})>\equiv  l<0> \oplus r< a(\{\emp\}) \oplus
b(\{\emp\})> &\text{($\mathit{Idempotency}$)}\\
\end{array}
\]
\end{example}
The relation $\equiv$ gives rise to the (surjective) equivalence map
$[-]\colon \Exp_{\F\lhd \G} \to \Exp_{\F\lhd \G}/_{\equiv}$ defined by $[\E] =
\{\E'\mid \E\equiv\E'\}$. The following diagram summarizes the maps we
have defined so far:
\[
\xymatrix{
\Exp_{\F\lhd \G}\ar[d]_{\delta_{\F\lhd \G}} \ar[r]^{[-]} &
\Exp_{\F\lhd \G}/_{\equiv}\\
\F(\Exp_{\F\lhd \G})\ar[r]_{\F([-])} & \F (\Exp_\G/_{\equiv}) 
}
\]
In order to complete the diagram, we next prove that
$\equiv$ is contained in the kernel
of $\F ([-]) \circ \delta_{\F\lhd \G}$. This will
guarantee the existence of a well-defined function 
$$\partial_{\F\lhd
\G}\colon
\Exp_{\F\lhd \G}/_{\equiv} \to \F(\Exp_\G/_{\equiv})$$ which, when $\F=\G$,
provides $\Exp_\G/_{\equiv}$ with a coalgebraic structure $\partial_\G\colon
\Exp_\G/_{\equiv}\to \G(\Exp_\G/_{\equiv})$ (as before, we write $\partial_\G$ to abbreviate
$\partial_{\G\lhd \G}$) and which makes $[-]$ a homomorphism of
coalgebras.
 
\begin{lemma}\label{lemma:h_welldef}
Let $\G$ and $\F$ be non\hyph deterministic functors, with $\F\lhd \G$. For all $\E_1,\E_2\in \Exp_{\F\lhd \G}$ with $\E_1\equiv\E_2$,
$$\F([-])\circ
\delta_{\F\lhd \G}(\E_1) = \F([-])\circ
\delta_{\F\lhd \G}(\E_2)$$
\end{lemma}

\proof
By induction on the length
of
derivations of $\equiv$.

First, let us consider derivations of length 1. We need to prove the
result for all the axioms in items $1.$ and $3.$ plus the axioms
$\mathit{FP}$ and $(\alpha-\mathit{equiv})$.

For the axioms in $1.$ the result follows by
Lemma~\ref{lemma:equivACIE}. The axiom
$\mathit{FP}$ follows trivially because of the definition
of $\delta_{\G}$, since $\delta_{\G}(\mu x . \gamma) =
\delta_{\G}(\gamma[\mu x.\gamma/x])$ and thus $\G([-]) \circ
\delta_\G(\mu x . \gamma)= \G([-])\circ \delta_\G(\gamma[\mu x.\gamma/x])$. 

\medskip

For the axiom $(\alpha-\mathit{equiv})$ we use the
$(\mathit{Cong})$ rule, which is proved below:
$$
\begin{array}{@{}l@{\ }ll}
&\G([-])\circ \delta_\G(\mu x.\gamma) \\
=& \G([-])\circ \delta_\G(\gamma[\mu
x.\gamma/x]) & \text{(def. of $\delta_\G$)}\\
=&  \G([-])\circ \delta_\G(\gamma[\mu
y.\gamma[y/x]/x])&\text{(by $(\mathit{Cong})$)}\\
=& \G([-])\circ \delta_\G(\gamma[y/x][\mu
y.\gamma[y/x]/y])& \text{\small($A[B[y/x]/x] = A[y/x][B[y/x]/y]$, $y$ not
free in $\gamma$)}\\
=& \G([-])\circ \delta_\G(\mu y. \gamma[y/x])) & \text{(def. of $\G([-])\circ \delta_\G$)}
\end{array}$$
\medskip

Let us next show the proof for some of the axioms in $3.$. The omitted
cases are similar. We show for each axiom $\E_1\equiv \E_2$ that
$\delta_{\F\lhd \G} (\E_1)=\delta_{\F\lhd \G} (\E_2)$.
\medskip

\begin{longtable}{@{}ll}
\fbox{$\bot_\B \equiv \emp$}&
\fbox{$b_1\oplus b_2 \equiv b_1\vee_\B b_2$}\\[2.2ex]
$
\delta_{\B\lhd \G}(\bot_\B)
= \bot_\B
= \delta_{\B\lhd \G}(\emp)
$
&
$\delta_{\B\lhd \G}(b_1\vee_\B b_2)
= b_1\vee_\B b_2
= \delta_{\B\lhd \G}(b_1\oplus b_2)
$
\\\\\\
\fbox{$l(\emp) \equiv \emp$}\\[2.2ex]
\multicolumn{2}{l}{
$
\delta_{\F_1\times \F_2\lhd \G}(l(\emp))
=
<\Empty_{\F_1\lhd \G}, \Empty_{\F_2\lhd \G}>
 = \delta_{\F_1\times \F_2\lhd \G}(\emp)
$}\\ \\\\
\fbox{$l(\E_1\oplus\E_2) \equiv l(\E_1)\oplus l(\E_2)$} 
\\[2ex]
\multicolumn{2}{l}{
$
\begin{array}{@{}c@{\ }l}
& \delta_{\F_1\times \F_2\lhd \G}(l(\E_1\oplus\E_2))\\
=&
<\delta_{\F_1\lhd \G}(\E_1\oplus\E_2), \Empty_{\F_2\lhd \G}>)\\
=&
<\Plus_{\F_1\lhd \G}(\delta_{\F_1\lhd \G}(\E_1),\delta_{\F_1\lhd
\G}(\E_2)
), \Plus_{\F_2\lhd \G}(\Empty_{\F_2\lhd \G},\Empty_{\F_2\lhd \G})>)\\
 =& \Plus_{\F_1\times \F_2}(<\delta_{\F_1\lhd \G}(\E_1),
\Empty_{\F_2\lhd \G}>,
<\delta_{\F_1\lhd \G}(\E_2), \Empty_{\F_2\lhd \G} >\\
=&
\delta_{\F_1\times \F_2\lhd \G}(l(\E_1)\oplus l(\E_2)))
\end{array}$}
\\\\
{\fbox{$l[\E_1\oplus\E_2] \equiv l[\E_1]\oplus
l[\E_2]$}} & \fbox{ $l[\E_1]\oplus r[\E_2] \equiv
l[\emp]\oplus r[\emp]$}
\\[2ex]$
\begin{array}{@{}c@{\ }l}
&\delta_{\F_1\myplus \F_2\lhd \G}(l[\E_1\oplus\E_2])\\
=
&\kappa_1(\delta_{\F_1\lhd \G}(\E_1\oplus\E_2))\\
 =& \Plus_{\F_1\myplus \F_2}(\kappa_1(\delta_{\F_1\lhd \G}(\E_1)),
\kappa_1(\delta_{\F_1\lhd \G}(\E_2))\\
=&
\delta_{\F_1\myplus \F_2\lhd \G}(l[\E_1]\oplus l[\E_2])
\end{array}$& $
\begin{array}{@{}c@{\ }l}
&\delta_{\F_1\myplus \F_2\lhd \G}(l[\E_1]\oplus r[\E_2])\\
=& \Plus_{\F_1\myplus \F_2}(\kappa_1(\delta_{\F_1\lhd \G}(\E_1)),
\kappa_2(\delta_{\F_2\lhd \G}(\E_2)))\\
=&\top\\
=& \Plus_{\F_1\myplus \F_2}(\kappa_1(\delta_{\F_1\lhd
\G}(\emp)),
\kappa_2(\delta_{\F_2\lhd \G}(\emp)))\\
=&
\delta_{\F_1\myplus \F_2\lhd \G}(l[\emp]\oplus r[\emp])
\end{array}$
\end{longtable}
Note that if we would have the axioms
$l[\emp] \equiv \emp$ and $r[\emp] \equiv \emp$ in
the axiomatization presented above, this theorem would not hold. 
$$
\begin{array}{l}
\delta_{\F_1\myplus \F_2\lhd \G}(l[\emp])
=\kappa_1([\bot])
\neq 
\bot 
=
\delta_{\F_1\myplus \F_2\lhd \G}(\emp)\\
\delta_{\F_1\myplus \F_2\lhd \G}(r[\emp])
=\kappa_2([\bot])
\neq 
\bot 
=
\delta_{\F_1\myplus \F_2\lhd \G}(\emp)

\end{array}
$$
Derivations with length $k\gr1$ can be obtained by two rules:
$(\mathit{Unique})$ or $(\mathit{Cong})$.
For the first (which uses the second), suppose that we have derived $\mu
x.\gamma \equiv \E$ and that we have already proved $\gamma[\E/x] \equiv
\E$. Then, we have:
\[
\begin{array}{@{}lcll}
\G([-])\circ \delta_\G(\mu x.\gamma) 
&=& \G([-])\circ \delta_\G(\gamma [\mu x.\gamma/x]) &\text{(def. $\delta_\G$)}\\
&=& \G([-])\circ \delta_\G(\gamma [\E /x]) &\text{(by $(\mathit{Cong})$)}\\
&{=}& \G([-])\circ \delta_\G(\E)&\text{(induction hypothesis)}
\end{array}
\]
For $(\mathit{Cong})$, suppose that we
have derived $\E[\E_1/x] \equiv \E[\E_2/x_2]$ and that we have already derived
$\E_1\equiv \E_2$, which gives us, as induction hypothesis, the equality
\begin{equation}\label{eq:ind}
(\F [-])(\delta_{\F\lhd \G}(\E_1))=(\F [-])(\delta_{\F\lhd \G}(\E_2))
\end{equation}
This equation is precisely what we need to prove the case  $\E=x$ (and
thus $\E_1,\E_2:\G\lhd \G$): 
\[
\begin{array}{lcll}
(\G [-])(\delta_{\G}(x[\E_1/x] ) &=& (\G[-])(\delta_{\G}(\E_1))\\
&=& (\G [-])(\delta_{\G}(\E_2)) &\text{(\ref{eq:ind})} \\
&=& (\G [-])(\delta_{\G}(x[\E_2/x]))
\end{array}
\]
For the cases $\E\neq x$, we prove that $\delta_{\F\lhd \G}(\E[\E_1/x]) =
\delta_{\F\lhd \G}(\E[\E_2/x])$, by induction on
the product of types of
expressions and expressions
(using the order defined in equation~(\ref{ind:order5})).
We show a few cases, the omitted ones are similar.
\[
\begin{array}{@{}lcl}
\delta_{\G\lhd \G}  ((\mu y. \E)[\E_1/x]) &=& \delta_{\G\lhd \G}
(\E[\E_1/x][\mu y.
\E/y]))\\
& \stackrel{\mathit{(IH)}}=& \delta_{\G\lhd \G}
(\E[\E_2/x][\mu y.
\E/y]))= \delta_{\G\lhd \G}  ((\mu y. \E)[\E_2/x])\\[1.1ex]
\delta_{\F_1\times \F_2 \lhd \G} (l<\E>[\E_1/x]) &=& <
\delta_{\F_1\lhd \G}(\E[\E_1/x]), \Empty_{\F_2\lhd \G} >\\&
\stackrel{\mathit{(IH)}}=& <
\delta_{\F_1\lhd \G}(\E[\E_2/x]), \Empty_{\F_2\lhd \G} >
=\delta_{\F_1\times \F_2 \lhd \G} (l<\E>[\E_2/x])\\[1.1ex]
\delta_{\F_1\myplus \F_2 \lhd \G} (l[\E][\E_1/x])& =&
\kappa_1(\delta_{\F_1
\lhd \G}(\E[\E_1/x]))\\&\stackrel{(\mathit{IH})}= &\kappa_1(\delta_{\F_1
\lhd \G}(\E[\E_2/x]))=\delta_{\F_1\myplus \F_2 \lhd \G} (l[\E][\E_2/x])
\end{array}
\]
\qedhere
Thus, we have a well-defined function $\partial_{\F \lhd
\G}\colon\Exp_{\F\lhd \G}/_{\equiv}\to
\F(\Exp_{\G}/_{\equiv})$, which makes the diagram above commute, that is $\partial_{\F\lhd \G}([\E]) = (\F[-])\circ
\delta_{\F\lhd \G}(\E)$. This provides the set
$\Exp_\G/_\equiv$ with a coalgebraic
structure $\partial_\G\colon \Exp_\G/_\equiv \to \G(\Exp_\G/_\equiv)$ which makes $[-]$ a
homomorphism between the coalgebras $(\Exp_\G, \delta_\G)$ and
$(\Exp_\G/_\equiv, \partial_\G)$. 
\subsection{Soundness and Completeness}\label{sec:sound_compl}
Next we show that the axiomatization introduced in the previous section is sound and complete.

Soundness is a direct consequence of the fact that the equivalence map $[-]$ is a coalgebra homomorphism.  

\begin{theorem}[Soundness]
Let $\G$ be a non\hyph deterministic functor. For all
$\E_1,\E_2\in \Exp_{\G}$,
\[
\E_1 \equiv \E_2 \Rightarrow \E_1 \sim \E_2
\]
\end{theorem}
\proof 
Let $\G$ be a non\hyph deterministic functor, let $\E_1,\E_2\in \Exp_{\G}$ and suppose that
$\E_1\equiv\E_2$. Then, $[\E_1] = [\E_2]$ and, thus
$$\mathbf{beh}_{\Exp_\G/_{\equiv}}([\E_1])=
\mathbf{beh}_{\Exp_\G/_{\equiv}}([\E_2])$$ where $\mathbf{beh}_S$
denotes, for any $\G$-coalgebra $(S,g)$, the unique map into the 
final coalgebra. The uniqueness of the map
into the final coalgebra and the fact that $[-]$ is
a coalgebra homomorphism implies that
$\mathbf{beh}_{\Exp_\G/_{\equiv}}\circ [-] = \mathbf{beh}_{\Exp_{\G}}$ which then yields $$
\mathbf{beh}_{\Exp_{\G}}(\E_1)
=  \mathbf{beh}_{\Exp_{\G}}(\E_2)$$ 
Since in the final
coalgebra
only the bisimilar elements are identified, $\E_1\sim\E_2$
follows.~
\qed

For completeness a bit more of work is required. Let us explain
upfront the key steps of the proof. The goal is to prove that 
$\E_1\sim \E_2 \Rightarrow \E_1\equiv \E_2$.
First, note that we have 
\begin{equation}\label{eq:comp}
\E_1\sim \E_2 \Leftrightarrow \mathbf{beh}_{\Exp_{\G}}(\E_1)
=  \mathbf{beh}_{\Exp_{\G}}(\E_2)\Leftrightarrow \mathbf{beh}_{\Exp_\G/_{\equiv}}([\E_1])=
\mathbf{beh}_{\Exp_\G/_{\equiv}}([\E_2])
\end{equation}
We then prove that $\mathbf{beh}_{\Exp_\G/_{\equiv}}$ is injective, which
is a sufficient condition to guarantee that $\E_1\equiv \E_2$ (since
it implies, together with~(\ref{eq:comp}), that $[\E_1]=[\E_2]$).

We proceed as follows. First, we factorize
 the map $\mathbf{beh}_{\Exp_\G/_{\equiv}}$ into an epimorphism followed by a
monomorphism~\cite[Theorem 7.1]{Rutten00} as shown in the following
diagram (where $I= \mathbf{beh}_{\Exp_\G/_{\equiv}}(\Exp_\G/_{\equiv})$):
\[
\xymatrix{
\Exp_\G/_{\equiv}\ar@{-->}@/{}^{1.4pc}/[rr]^-{\mathbf{beh}_{\Exp_\G/_{\equiv}}}
\ar@{>>}[r]^-{e}
\ar[d]_{\partial_{\G}} &
I\ar@{^{(}->}[r]^-{m} \ar[d]_{\overline\omega_\G} &
\Omega_\G\ar[d]^{\omega_\G}\\
\G(\Exp_\G/_{\equiv})\ar[r]& \G(I)\ar[r]
& \G(\Omega_\G)\\
}
\]
Then, we prove that (1) $(\Exp_\G/_{\equiv}, \partial_{\G})$ is a locally
finite coalgebra (Lemma~\ref{lemma:h_locally_finite}) and (2) both
coalgebras $(\Exp_\G/_{\equiv}, \partial_{\G})$ and $(I,
\overline\omega_\G)$ are final in the category of locally finite $\G$-coalgebras
(Lemmas~\ref{lemma:final1} and \ref{lemma:final2}, respectively).
Since final coalgebras are unique up to isomorphism, it follows that
$e\colon \Exp_\G/_{\equiv}\to I$
is in fact an isomorphism and therefore $\mathbf{beh}_{\Exp_\G/_{\equiv}}$ is injective, which will give us completeness.

In the case of the deterministic
automata functor $\D=2\times \id^A$, the set $I$ will be precisely the
set of regular languages, the class of languages that can be denoted
by regular expressions. This means that final locally finite
coalgebras generalize regular languages (in the same way that final coalgebras
generalize the set of all languages). 

We proceed with presenting and proving the extra lemmas needed in
order to prove completeness. We start by showing that the coalgebra 
$(\Exp_\G/_{\equiv}, \partial_{\G})$ is locally finite (note
that this implies that $(I,
\overline\omega_\G)$ is also locally finite) and that $\partial_{\G}$ is
an isomorphism. 
\begin{lemma}\label{lemma:h_locally_finite}
The coalgebra $(\Exp_\G/_{\equiv}, \partial_{\G})$ is a locally finite coalgebra.
Moreover, $\partial_{\G}$ is an isomorphism.
\end{lemma}
\proof 
Local finiteness
is a direct consequence of the generalized Kleene's theorem
(Theorem~\ref{kleene2}). In the proof of Theorem~\ref{kleene2} we showed that, 
given $\E\in \Exp_{\G}$, the
subcoalgebra $<[\E]_{\mathit{ACIE}}>$ is finite. Thus, the
subcoalgebra $<[\E]>$ is also finite (since $\Exp_\G/_\equiv$ is a
quotient of $\Exp_\G/_{\equiv_{ACIE}}$).

To see that $\partial_{\G}$ is an isomorphism,
first define, for every $\F\lhd \G$,
\begin{equation}\label{eq:h_inv}
\partial^{-1}_{\F\lhd \G} (c) = [\overline\gamma^\F_{c}] 
\end{equation}
where $\overline\gamma^\F_{c}$ is defined, for $\F\neq \id$, as
$\gamma^\F_{c}$ in the proof of Theorem~\ref{kleene1}, and for $\F=\id$
as $\overline\gamma^{\id}_{[\E]} = \E$. Then, we prove that
$\partial^{-1}_{\F\lhd \G}$ has
indeed the properties \mycirc{1} 
$\partial^{-1}_{\F\lhd \G} \circ \partial_{\F\lhd \G}  = \mathit{id}_{\Exp_{\F\lhd \G}/_\equiv}$ and \mycirc{2} 
$\partial_{\F\lhd \G} \circ
\partial^{-1}_{\F\lhd \G} = \mathit{id}_{\F(\Exp_{\F\lhd \G}/_\equiv)}$.
Instantiating $\F=\G$ one derives that $\delta_\G$ is an isomorphism. It is enough to prove for 
\mycirc{1} that $\overline\gamma^\F_{\partial_{\F\lhd \G}([\E])}\equiv \E$ and
for \mycirc{2} that $\partial_{\F\lhd \G}([\overline\gamma^\F_{c}])=c$. We illustrate a few cases. The omitted ones are similar.

\medskip

\noindent\mycirc{1} By induction on the product of types of
expressions and expressions
(using the order defined in equation~(\ref{ind:order5})).
$$
\begin{array}{@{}l}
\overline\gamma^\id_{\partial_{\id\lhd \G}([\E])} = \E\\[1.5ex]
 \overline\gamma^{\F_1\times \F_2}_{\partial_{\F_1\times \F_2\lhd \G}([r<\E>])} =
l<\overline\gamma^{\F_1}_{\partial_{\F_1\lhd \G}(\emp)}>\oplus
r<\overline\gamma^{\F_2}_{\partial_{\F_2\lhd \G}(\E)}>
\stackrel{(\mathit{IH})}\equiv l<\emp>\oplus
r<\E> \equiv r<\E>\\[1.5ex]
\overline\gamma^\G_{\partial_{\G}([\mu x.\E])} =
\overline\gamma^\G_{\partial_{\G}([\E[\mu x.\E/x]])}
\stackrel{(\mathit{IH})}\equiv
\E[\mu x.\E/x]\equiv \mu x.\E
\end{array}
$$
Note that the cases $\E=\emp$ and 
$\E=\E_1\oplus\E_2$ require an extra proof (by induction on $\F$).
More precisely, one needs to prove that 
$$
\mycirc{$a$}\ \ 
 \overline\gamma^\F_{\F[-](\Empty_{\F\lhd \G})} \equiv \emp
\text{ and } \mycirc{$b$}\ \  \overline\gamma^\F_{\F[-](\Plus_{\F\lhd \G}(x_1,x_2))} \equiv \overline\gamma^\F_{\F[-](x_1)}\oplus
\overline\gamma^\F_{\F[-](x_2)}$$ It is an easy proof by induction. We
illustrate here only the cases $\F=\id$, $\F=\B$ and $\F=\F_1\times \F_2$. 
\begin{longtable}{@{}lll}
$\mycirc{$a$}$&$\overline\gamma^\id_{[\emp]} = \emp$\\[1.2ex]
&$\overline\gamma^\B_{[\bot_\B]} = \bot_\B \equiv \emp$\\[1.2ex]
&$\overline\gamma^{\F_1\times \F_2}_{<\F_1[-](\Empty_{\F_1\lhd \G}),
\F_2[-](\Empty_{\F_2\lhd \G})>}$& $=
l<\overline\gamma^{\F_1}_{\F_1[-](\Empty_{\F_1\lhd \G})}>\oplus
r<\overline\gamma^{\F_2}_{\F_2[-](\Empty_{\F_2\lhd \G})}>$\\
&&$\stackrel{(\mathit{IH})}\equiv l<\emp>\oplus  r<\emp> \equiv
\emp$\\[2ex]
$\mycirc{$b$}$ & $\overline\gamma^\id_{[x_1\oplus x_2]} = x_1\oplus x_2 =
\overline\gamma^\id_{[x_1]}\oplus \overline\gamma^\id_{[x_2]}$\\[1.4ex]
&$\overline\gamma^\B_{[x_1\vee_\B x_2]} = x_1\vee_\B x_2 \equiv  x_1\oplus
x_2$ &$= \overline\gamma^\B_{[x_1]}\oplus \overline\gamma^\B_{[x_2]}$\\[1.4ex]
&$\overline\gamma^{\F_1\times \F_2}_{{\F_1\times
\F_2}[-](\Plus_{{\F_1\times \F_2}\lhd \G}(<u_1,v_1>,<u_2,v_2>))}$
&$\stackrel{\phantom{(\mathit{IH})}}= 
\overline\gamma^{\F_1\times \F_2}_{<\Plus_{\F_1}(u_1,v_1),
\Plus_{\F_2}(u_2,v_2)>}$ \\&&$\stackrel{\phantom{(\mathit{IH})}}=
l<\overline\gamma^{\F_1}_{\Plus_{\F_1}(u_1,v_1)}>\oplus
r<\overline\gamma^{\F_2}_{\Plus_{\F_2}(u_2,v_2)}>$
\\[1.2ex]
&&$
\stackrel{(\mathit{IH})}\equiv
l<\gamma^{\F_1}_{u_1}\oplus\gamma^{\F_1}_{v_1}>\oplus
r<\gamma^{\F_2}_{u_2}\oplus \gamma^{\F_2}_{v_2}>$
\\[1.2ex]
&&$\stackrel{\phantom{(\mathit{IH})}}\equiv
(l<\gamma^{\F_1}_{u_1}>\oplus r<\gamma^{\F_2}_{u_2}>) \oplus
(l<\gamma^{\F_1}_{v_1}>\oplus r<\gamma^{\F_2}_{v_2}>)$ \\[1.2ex]
&&$\stackrel{\phantom{(\mathit{IH})}}=
\overline\gamma^{\F_1\times \F_2}_{<u_1,u_2>} \oplus
\overline\gamma^{\F_1\times \F_2}_{<v_1,v_2>}$
\end{longtable}
\noindent\mycirc{2} By induction on the structure of $\F$.
\[
\begin{array}{lcl}
\partial_{\F_1\myplus \F_2 \lhd \G}([\overline\gamma^{\F_1\myplus \F_2}_{c}])&=&
\begin{cases} \partial_{\F_1\myplus \F_2\lhd
\G}([l[\overline\gamma^{\F_1}_{c'}]])= \kappa_1(\partial_{\F_1\lhd
\G}([\overline\gamma^{\F_1}_{c'}]))&c= \kappa_1(c')\\
\partial_{\F_1\myplus \F_2\lhd
\G}([r[\overline\gamma^{\F_2}_{c'}]]) = \kappa_2(\partial_{\F_2\lhd
\G}([\overline\gamma^{\F_2}_{c'}]))
&c=\kappa_2(c') \\ \partial_{\F_1\myplus \F_2\lhd \G}([\emp])=\bot &c=\bot \\
\partial_{\F_1\myplus \F_2 \lhd
\G}([l[\emp]\oplus r[\emp]])=\top & c=\top \end{cases}\\
&\stackrel {(\mathit{IH})}=& c
\end{array}
\]
\[
\begin{array}{lcl}
\partial_{\pow \F \lhd \G}([\overline\gamma^{\pow \F}_{C}])&=&
\begin{cases} \partial_{\pow \F \lhd
\G}([\emp])= \emptyset & C=\emptyset\\
\partial_{\pow \F \lhd
\G}([\bigoplus_{c\in C}\overline\gamma^{\F_1}_{c}])=\{\partial_{\F \lhd
\G}([\overline\gamma^{\F_1}_{c}])\mid c\in C\} &
\text{otherwise}\end{cases} \\
&\stackrel {(\mathit{IH})}=& C
\end{array}
\]
\qed

Next, we prove the analogue of the following useful and intuitive
equality on regular expressions. Given a deterministic
automaton $<o,t>\colon S\to 2\times S^A$ and a state $s\in S$, the
associated regular expression $r_s$  can be written as
\begin{equation}\label{eq:regexp_decomp}
r_s = o(s) + \sum_{a\in A} a\cdot r_{t(s)(a)}
\end{equation}
using the axioms of Kleene algebra~\cite[Theorem 4.4]{Brz64}.
\begin{lemma}\label{lemma:gamma}
Let $(S,g)$ be a locally finite $\G$-coalgebra, with $\G\neq \id$, and let $s\in S$, with
$<s>=\{s_1,\ldots,s_n\}$ (where $s_1=s$). Then:
\begin{eqnarray}\label{gamma_eq} 
\expr{s_i} &\equiv& \gamma_{g(s_i)}^\G \{\expr {s_1}/x_1\} \ldots \{\expr
{s_n}/x_n\}
\end{eqnarray}
\end{lemma}
\proof 
Let $A_i^k$, where $i$ and $k$ range from $1$ to $n$, be the terms
defined as in the proof of Theorem~\ref{kleene1}. Recall that
$\expr{s_i}=A^n_i$. We calculate:
\[
\begin{array}{@{}l@{\ }l@{\ \ }l@{}}
& \expr{ s_i }\\ = & A^n_i \\
=& (\mu x_i. \gamma^\G_{g(s_i)}) \{A^0_1/x_1\} \ldots
\{A_n^{n-1}/x_n\}\\[1ex] 
=& \mu x_i. (\gamma^\G_{g(s_i)} \{A^0_1/x_1\} \ldots
\{A_{i-1}^{i-2}/x_{i-2}\}\{A^i_{i+1}/x_{i+1}\} \ldots \{A_n^{n-1}/x_n\}
)\\[1ex] 
\equiv& \gamma^\G_{g(s_i)} \{A^0_1/x_1\} \ldots
\{A_{i-1}^{i-2}/x_{i-2}\}\{A^i_{i+1}/x_{i+1}\} \ldots \{A_n^{n-1}/x_n\}
\{A^n_i/x_i\} &\text{(fixed point axiom\footnotemark )}\\[1ex] 
=& \gamma^\G_{g(s_i)} \{A^0_1/x_1\} \ldots \{A_n^{n-1}/x_n\} &  \text{(by
\ref{eq:subst1})}\\[1ex] 
=&  \gamma^\G_{g(s_i)} \{A^0_1\{A^1_2/x_2\}\ldots \{A_n^{n-1}/x_n\}/x_1\} \ldots \{A_n^{n-1}/x_n\} &  \text{(by
\ref{eq:subst2})}\\[1ex] 
=&  \gamma^\G_{g(s_i)} \{A^n_1/x_1\}\{A_2^1/x_2\}
\ldots \{A_n^{n-1}/x_n\} &  \text{(def. $A_1^{n}$)}\\[2ex] 
\vdots& \text{(repeat last 2 steps for $A_2^1,\ldots,
A_{n-1}^{n-2}$)}\\[2ex]
=&  \gamma^\G_{g(s_i)} \{A^n_1/x_1\}\{A_2^n/x_2\}
\ldots \{A_n^n/x_n\} &  \text{($A_{n-1}^{n} = A^n_n$)}\\
\end{array}
\]
\footnotetext{Note that the fixed point axiom can be formulated using syntactic replacement rather than substitution -- $\gamma\{\mu x.\gamma/x\} \equiv \mu x.\gamma$ -- since $\mu x . \gamma$ is a closed term.}
\qed

Instantiating~(\ref{gamma_eq}) for $<o,t>\colon S\to
2\times S^A$, one can easily spot the similarity with
equation~(\ref{eq:regexp_decomp}) above:
\[
\expr s \equiv l<o(s)> \oplus r\Big<\bigoplus_{a\in A} a\big(\expr{
t(s)(a)}\big)\Big>
\]

Next, we prove that there exists a coalgebra homomorphism between any
locally finite $\G$- coalgebra $(S,g)$ and $(\Exp_\G/_{\equiv}, \partial_{\G})$.

\begin{lemma}\label{lemma:ceil-exists}
Let $(S,g)$ be a locally finite $\G$-coalgebra. There exists a coalgebra
homomorphism $\ceil{-}
\colon S \to \Exp_\G/_{\equiv}$.
\end{lemma}

\proof 
We define $\ceil{-}= [-]\circ \expr{-}$, where $\expr{-}$ is as
 in the proof of Theorem~\ref{kleene1}, associating to a state $s$ of
a locally finite coalgebra an expression $\expr s$ with $s\sim \expr s$. To prove that $\ceil{-}$ is a
homomorphism we need to verify that $(\G\ceil -)\circ g = \partial_\G\circ
\ceil -$.

If $\G=\id$, then  $(\G\ceil -)\circ g (s_i) = [\emp] = \partial_\G (\ceil
{s_i})$. For $\G\neq \id$ we calculate, using Lemma~\ref{lemma:gamma}:
\begin{eqnarray*}
\partial_\G(\ceil {s_i})
= \partial_\G([\gamma_{g(s_i)}^\G [\expr {s_1}/x_1] \ldots [\expr
{s_n}/x_n]]) 
\end{eqnarray*}
and we then prove the more general equality, for $\F\lhd \G$ and
$c\in \F<s>$:
\begin{equation}\label{eq:interm}
\partial_{\F\lhd \G} ([\gamma_{c}^\F [\expr {s_1}/x_1] \ldots [\expr
{s_n}/x_n]]) = \F \ceil - (c)
\end{equation}
The intended equality then follows by taking $\F=\G$ and $c=g(s_i)$. 
Let us prove the equation (\ref{eq:interm}) by induction on $\F$.

\medskip 
\fbox{$\F=\id$} $c= s_j \in <s>$
\begin{eqnarray*}
\partial_{\id\lhd \G} ([\gamma_{s_j}^{\id} [\expr {s_1}/x_1] \ldots [\expr
{s_n}/x_n]])= [\expr {s_j}]
 =  \ceil{s_j}
\end{eqnarray*}

\fbox{$\F=\B$} $c= b \in \B$
\begin{eqnarray*}
\partial_{\B\lhd \G} ([\gamma_{b}^{\B} [\expr {s_1}/x_1] \ldots [\expr
{s_n}/x_n]]) = [b]
 =  \B \ceil{-} (b)
\end{eqnarray*}

\fbox{$\F=\F_1\times \F_2$} $c=<c_1,c_2> \in (\F_1\times \F_2) <s>$

\begin{eqnarray*}
&&\partial_{\F_1\times \F_2\lhd \G} ([\gamma_{<c_1,c_2>}^{\F_1\times \F_2} [\expr {s_1}/x_1] \ldots [\expr
{s_n}/x_n]])\\ &=& \partial_{\F_1\times \F_2\lhd \G} ([l(\gamma_{c_1}^{\F_1})\oplus
r(\gamma_{c_2}^{\F_2})[\expr {s_1}/x_1] \ldots [\expr
{s_n}/x_n]])\\
&=& < \partial_{\F_1\lhd \G} ([\gamma_{c_1}^{\F_1}[\expr {s_1}/x_1] \ldots [\expr
{s_n}/x_n]]), \partial_{\F_2\lhd \G}
([\gamma_{c_2}^{\F_2}[\expr {s_1}/x_1] \ldots [\expr
{s_n}/x_n]]) >\\
& \stackrel{(\mathit{IH})}= & <\F_1 \ceil - (c_1), \F_2 \ceil - (c_2)>\\
&=& (\F_1\times \F_2 \ceil -) (c)
\end{eqnarray*}

\fbox{$\F=\F_1\myplus \F_2$}, 
\fbox{$\F=\F_1^A$} and 
\fbox{$\F=\pow \F_1$}: similar to $\F_1\times \F_2$.
\qedhere\medskip

\noindent We can now prove that the coalgebras $(\Exp_\G/_{\equiv},
\partial_{\G})$ and $(I,\overline \omega_\G)$ are both final in the
category of locally finite $\G$-coalgebras.
\begin{lemma}\label{lemma:final1}
The coalgebra $(I,\overline \omega_\G)$ is final in the category
$\mathit{Coalg}(\G)_{{\mathbf{LF}}}$.
\end{lemma}
\vfill\eject
\proof 
We want to show that for any locally finite $\G$-coalgebra $(S,g)$, there exists a {\em unique} 
homomorphism~$(S,g) \to (I,\overline \omega_\G)$. The existence is
guaranteed by Lemma~\ref{lemma:ceil-exists}, where the homomorphism {$\ceil{-}\colon S
\to \Exp_\G/_{\equiv} $} is defined. Post-composing this homomorphism with
$e$ (defined above) we get a coalgebra homomorphism $e\circ \ceil - \colon S \to I$.
If there is another homomorphism $f\colon S\to I$, then by post-composition with the inclusion 
$m\colon I\hookrightarrow \Omega$ we get two homomorphisms ($m\circ f$
and $m\circ e\circ \ceil -$) into the
final $\G$-coalgebra. Thus, $f$ and $e\circ \ceil -$ must be equal.
\qed

\begin{lemma}\label{lemma:final2}
The coalgebra $(\Exp_\G/_{\equiv}, \partial_{\G})$ is final in the category
$\mathit{Coalg}(\G)_{\mathbf{LF}}$.
\end{lemma}
\proof 
We want to show that for any locally finite $\G$-coalgebra $(S,g)$, there exists a {\em unique} 
homomorphism~$(S,g) \to (\Exp_\G/_{\equiv}, \partial_{\G})$. We only
need to prove uniqueness, since the existence is guaranteed by
Lemma~\ref{lemma:ceil-exists}, where {$\ceil{-}\colon S \to \Exp_\G/_{\equiv} $} is defined. 

Suppose we have
another homomorphism $f\colon S\to \Exp_\G/_{\equiv}$. Then, we
shall prove that $f=\ceil -$. Let, for any $s\in S$, $f_s$ denote any 
representative of $f(s)$ (that is, $f(s)=[f_s]$). First, observe that 
because $f$ is a homomorphism the following holds for every $s\in S$:
\begin{equation}\label{eq:m}
f(s) = (\partial_{\G}^{-1} \circ \G(f) \circ g)(s) \Leftrightarrow
f_s\equiv  \gamma_{g(s)}^\G [f_{s_1}/x_1]\ldots [f_{s_n}/x_n]
\end{equation}
where $<s> = \{s_1,\ldots,s_n\}$, with $s_1=s$ (recall that
$\partial^{-1}_{\G}$ was defined in~(\ref{eq:h_inv}) and note that
$\overline\gamma^\G_{(\G(f) \circ g)(s)} = \gamma^\G_{g(s)} [f_{s_i}/x_{i}]$).

Next, we prove that $f_{s_i}\equiv \expr{s_i}$ (which is equivalent to $f(s_i) = \ceil {s_i}$), for all
$i=1,\ldots n$. For simplicity we will here prove the case $n=3$. The general case is identical but notationally heavier.
First, we prove that $f_{s_1} \equiv A_1[f_{s_2}/x_2][f_{s_3}/x_3]$.
\[
\begin{array}{@{}lll}
&f_{s_1} \equiv \gamma_{g(s_1)}^\G [f_{s_1}/x_1]  [f_{s_2}/x_2]  [f_{s_3}/x_3] &\text{(by~(\ref{eq:m})})\\[.8ex]
\Leftrightarrow& f_{s_1} \equiv \gamma_{g(s_1)}^\G [f_{s_2}/x_2]  [f_{s_3}/x_3] [f_{s_1}/x_1]  &\text{(all $f(s_i)$ are closed)}\\[.8ex]
\Rightarrow& f_{s_1} \equiv \mu x_1. \gamma_{g(s_1)}^\G [f_{s_2}/x_2]  [f_{s_3}/x_3]  &\text{(by uniqueness of fixed points)}\\[.8ex]
\Leftrightarrow& f_{s_1} \equiv A_1[f_{s_2}/x_2]  [f_{s_3}/x_3]  &\text{(def. of $A_1$)}
\end{array}
\]
Now, using what we have computed for $f_{s_1}$ we prove that $f_{s_2} \equiv A_2^1[f_{s_3}/x_3]$.
{
\[
\begin{array}{@{}lll} 
&f_{s_2} \equiv \gamma_{g(s_2)}^\G [f_{s_1}/x_1]  [f_{s_2}/x_2]  [f_{s_3}/x_3] &\text{(by~(\ref{eq:m})})\\[.8ex]
\Leftrightarrow&f_{s_2} \equiv \gamma_{g(s_2)}^\G [A_1/x_1]  [f_{s_2}/x_2]  [f_{s_3}/x_3]
&\text{(expressions for $f_{s_1}$ and (\ref{eq:subst2}))}\\[.8ex]
\Leftrightarrow& f_{s_2} \equiv \gamma_{g(s_2)}^\G [A_1/x_1]  [f_{s_3}/x_3] [f_{s_2}/x_2] &\text{(all $f(s_i)$ are closed)}\\[.8ex]
\Rightarrow& f_{s_2} \equiv \mu x_2. \gamma_{g(s_2)}^\G [A_1/x_1]  [f_{s_3}/x_3]  &\text{(by uniqueness of fixed points)}\\[.8ex]
\Leftrightarrow& f_{s_2} \equiv A_2^1[f_{s_3}/x_3]  &\text{(def. of $A_2^1$)}
\end{array}
\]}
At this point we substitute $f_{s_2}$ in the expression for $f_{s_1}$ by $A_2^1[f_{s_3}/x_3]$ which yields:
\[
f_{s_1} \equiv  A_1[A^1_2 [f_{s_3}/x_3]  /x_2]  [f_{s_3}/x_3]  \equiv A_1[A^1_2/x_2]  [f_{s_3}/x_3] 
\]
Finally, we prove that $f_{s_3} \equiv A_3^2$:
\[
\begin{array}{@{}lll} 
&f_{s_3} \equiv \gamma_{g(s_3)}^\G [f_{s_1}/x_1]  [f_{s_2}/x_2]  [f_{s_3}/x_3] &\text{(by~(\ref{eq:m})})\\[.8ex]
\Leftrightarrow&f_{s_3} \equiv \gamma_{g(s_3)}^\G [A_1/x_1]  [A_2^1/x_2]  [f_{s_3}/x_3]
&\text{(expr. for $f(s_i)$ and (\ref{eq:subst2}))}\\[.8ex]
\Rightarrow& f_{s_3} \equiv \mu x_3. \gamma_{g(s_3)}^\G [A_1/x_1]  [A_2^1/x_2]  &\text{(by uniqueness of fixed points)}\\[.8ex]
\Leftrightarrow& f_{s_3} \equiv A_3^2  &\text{(def. of $A_3^2$)}
\end{array}
\]
Thus, we have $f_{s_1} \equiv A_1[A^1_2/x_2]  [A^2_3/x_3] $, $f_{s_2}
\equiv A_2^1  [A^2_3/x_3]$ and $f_{s_3} \equiv A^2_3$. Note that
$A_2^1  [A^2_3/x_3]  \equiv A_2^1  \{A^2_3/x_3\} $ since $x_2$ is not
free in $A^2_3$. Similarly, since $x_1$ is not free in $A^1_2$ and
$A^2_3$, we have that $A_1[A^1_2/x_2]  [A^2_3/x_3]  \equiv A_1\{A^1_2/x_2\}  \{A^2_3/x_3\}$. Thus $f(s_i) = \ceil{s_i}$, for all $i=1,2,3$.
\qed

As a consequence of
Lemma~\ref{lemma:final2}, we have that if $\G_1$ and $\G_2$ are
isomorphic
functors then $\Exp_{\G_1}/_\equiv$ and $\Exp_{\G_2}/_\equiv$ are
also isomorphic (for instance, this would be true for $\G_1(X) = \B
\times (X \times A)$ and
$\G_2(X) =  A \times (\B \times X)$). 

We remark that Lemma~\ref{lemma:final1} could have been proved as a
consequence of Lemma~\ref{lemma:final2}, by observing that
$(I,\overline \omega_\G)$ is, by construction, a quotient of
$(\Exp_\G/_{\equiv}, \partial_{\G})$.  

At this point, because final objects are unique up-to isomorphism, we
know that $e\colon \Exp_\G/_\equiv \to I$ is an isomorphism and hence we
can conclude that the map $\mathbf{beh}_{\Exp_\G/_{\equiv}}$ is injective, since it factorizes into
an isomorphism followed by a mono. This fact is the last
thing we need to prove completeness.

\begin{theorem}[Completeness]
Let $\G$ be a non\hyph deterministic functor. For all
$\E_1,\E_2\in \Exp_{\G}$,
\[
\E_1 \sim \E_2 \Rightarrow \E_1 \equiv \E_2
\]
\end{theorem}

\proof 
Let $\G$ be a non\hyph deterministic functor, let $\E_1,\E_2\in \Exp_{\G}$ and suppose that
 $\E_1\sim\E_2$.
Because only bisimilar elements are identified in the final coalgebra
we know that it must be the case that $\mathbf{beh}_{\Exp_\G}(\E_1)
=  \mathbf{beh}_{\Exp_\G}(\E_2) $ and thus, since the equivalence class
map $[-]$ is a homomorphism, $\mathbf{beh}_{\Exp_\G/_{\equiv}}([\E_1])=
\mathbf{beh}_{\Exp_\G/_{\equiv}}([\E_2])$. Because
$\mathbf{beh}_{\Exp_\G/_{\equiv}}$ is injective we have that $[\E_1]=
[\E_2]$. Hence, $\E_1\equiv\E_2$.
\qed

\section{Two more examples}\label{sec:appl}

In this section we apply our framework to two other examples: 
labelled transition systems (with explicit termination) and automata on guarded strings.
These two automata models are directly connected to, respectively, basic process algebra and 
Kleene algebra with tests. To improve readability we will present the corresponding languages 
using a more user-friendly syntax than the canonically derived one.

\paragraph{\textbf{Labelled transition systems.}} Labelled transition systems (with explicit termination) are coalgebras
for the functor $1+(\pow \id)^A$.  As we will show below, instantiating our framework for this functor produces a language that
is equivalent to the closed and guarded expressions generated by the
following grammar, where $a\in A$ and $x\in X$ ($X$ is a set of
fixed point variables):
\[
P\, ::= \mathbf{0} \mid P+P \mid a.P \mid \delta\mid \surd \mid \mu x. P
\mid x
\]
together with the equations (omitting the congruence and $\alpha$-equivalence rules)

\begin{align*}
&P_1 + P_2 \equiv P_2+P_1 && P_1 + (P_2+P_3) \equiv (P_1+P_2)+P_3\\
&P+P  \equiv P &&  P+ \mathbf{0} \equiv P\\
&P+ \delta \equiv P\ (\star)&
&\surd + \delta \equiv \surd + P\ (\star)& (\star) \text{ if }  P\not\equiv
\mathbf{0}\text{ and } P\not\equiv \surd\\
&P[\mu x.P/x] \equiv \mu x . P && P[Q/x]  \equiv  Q \Rightarrow (\mu
x.P)  \equiv Q
\end{align*}
Note that, as expected, there is no law that allows us to prove $a.(P+Q)\equiv a.P+a.Q$. Moreover, 
 observe that this syntax and axiomatization is very similar to the one presented 
in~\cite{AcetoH92}. In the syntax above, $\delta$
represents deadlock, $\surd$ successful termination and $\mathbf 0$
the totally undefined process.

We will next show how the beautified syntax above was derived from the
canonically derived syntax for the expressions $\E\in \Exp_{1+(\pow
\id)^A}$, which is given by the set of closed and guarded expressions
defined by the following BNF:
\[
\begin{array}{l}
\E ::= \emp \mid \E\oplus\E \mid x \mid \mu x.\E \mid l[\E_1]
\mid r[\E_2]\\
\E_1 ::= \emp \mid \E_1\oplus\E_1 \mid *\\
\E_1 ::= \emp \mid \E_2\oplus\E_2 \mid a(\E')\\
\E' ::= \emp \mid \E'\oplus\E' \mid \{\E\}
\end{array}
\]
We define two maps between this grammar and the grammar presented above.
Let us start to show how to translate $P$'s into $\E$'s, by defining 
a map $(-)^\dagger$ by induction on the structure of $P$:
$$
\begin{array}{lcl}
(\mathbf{0})^\dagger &=& \emp\\
(P_1+P_2)^\dagger &=&  (P_1)^\dagger \oplus (P_2)^\dagger\\
(\mu x. P)^\dagger &=& \mu x. P^\dagger\\
x^\dagger  &=& x
\end{array}
\begin{array}{lcl}
(a.P)^\dagger &=& r[a(\{P^\dagger\})]\\ 
(\surd)^\dagger &=& l[*] \\
(\delta)^\dagger &=& r[\emp]\\
\end{array}
$$
And now the converse translation:
$$
\begin{array}{lcl}
(\emp)^\ddagger &=& \mathbf{0}\\
(\E_1\oplus\E_2)^\ddagger &=&  (\E_1)^\ddagger + (\E_2)^\ddagger\\
(\mu x. \E)^\ddagger &=& \mu x. \E^\ddagger\\
x^\ddagger  &=& x\\
(l[\emp])^\ddagger &=& \surd \\
(l[\E_1\oplus\E_1'])^\ddagger &=& (l[\E_1])^\ddagger + 
(l[\E_1'])^\ddagger\\
\end{array}
\begin{array}{lcl}
(l[*])^\ddagger &=& \surd \\
(r[\emp])^\ddagger &=& \delta \\ 
(r[\E_2\oplus\E_2'])^\ddagger &=&(r[\E_2])^\ddagger  + (r[\E_2'])^\ddagger\\ 
(r[a(\emp)])^\ddagger &=& \delta \\ 
(r[a(\E_1'\oplus\E_2')])^\ddagger &=&(r[a(\E_1')])^\ddagger + (r[a(\E_2')])^\ddagger \\ 
(r[a(\{\E\})])^\ddagger &=& a.\E^\ddagger \\ 
\end{array}
$$
One can prove that if $P_1\equiv P_2$ (using the equations above) 
then $(P_1)^\dagger\equiv (P_2)^\dagger$  (using the automatically derived equations 
for the functor) and also that  $\E_1\equiv \E_2$ implies $(\E_1)^\ddagger \equiv (\E_2)^\ddagger$. 

\paragraph{\textbf{Automata on guarded strings.}} It has recently been shown~\cite{kozen08} that automata on guarded
strings (acceptors of the join irreducible elements of the free Kleene algebra
with tests on generators $\Sigma, T$) are coalgebras for the functor
$\B\times \id^{At\times \Sigma}$, where $At$ is the set of 
atoms, {\em i.e.} minimal nonzero elements of the free Boolean algebra $\B$
generated by $T$ and $\Sigma$ is a set of actions.  Applying our framework to this functor yields a
language that is equivalent to the closed and guarded expressions
generated by the following grammar, where $b\in \B$ and $a\in \Sigma$:
\[
P\, ::= \mathbf{0} \mid <b> \mid P+P \mid b \to a.P \mid \mu x. P
\mid x
\]
accompanied by the equations (omitting the congruence and $\alpha$-equivalence rules)
\begin{align*}
&P_1 + P_2 \equiv P_2+P_1 && P_1 + (P_2+P_3) \equiv (P_1+P_2)+P_3\\
&P+P \equiv P&&P+ \mathbf{0} \equiv P \\ 
&<b_1>+<b_2> \equiv <b_1\vee_\B b_2> && \mathbf{0} \equiv <\bot_\B>\\
& (b \to a.\mathbf{0})\equiv  \mathbf{0} && (\bot_\B \to a.P) 
\equiv \mathbf{0}\\
& (b \to a.P_2) +  (b \to a.P_2) 
\equiv  b \to a.(P_1+P_2)&&
 (b_1 \to a.P) +  (b_2 \to a.P) \equiv  (b_1\vee_\B b_2) \to a.P\\
&P[\mu x.P/x] \equiv \mu x . P&&
P[Q/x] \equiv  Q \Rightarrow (\mu
x.P) \equiv Q
\end{align*}
We will not present a full comparison of this syntax to the one of Kleene algebra 
with tests~\cite{kozen08} (and propositional Hoare triples). The
differences between our syntax and that of KAT are similar to the ones
between regular expressions and the language $\Exp_\D$ for the functor
representing deterministic automata (see
Definition~\ref{def:regexp_to_exp}). 
Similarly to the LTS example one can define maps between the beautified syntax and the automatically generated one and prove its correctness.

\section{Polynomial and finitary coalgebras}\label{sec:pol_fin}

The functors we considered above allowed us to modularly derive
languages and axiomatizations for a large class of coalgebras. If we
consider the subset of $\ndf$ without the $\pow$ functor, the class of
coalgebras for these functors almost coincides with polynomial
coalgebras (that is, coalgebras for a polynomial functor). The only
difference comes from the use of join-semilattices for constant
functors and $\myplus$ instead of the ordinary
coproduct, which played an important role in order for
us to be able to have underspecification and overspecification.  
We will next show how to derive expressions and axiomatizations
directly for polynomial coalgebras, where no underspecification or
overspecification is allowed. 

Before we show the formal definition, let us provide some intuition. 
The main changes\footnote{This syntax was suggested to us by B. Klin,
during CONCUR'09.}, compared to the previous sections, would be not to 
have $\emp$ and $\oplus$ and consider an expression $<-,->$ for the
product instead of the two expressions $l<->$ and $r<->$ which we
considered and an expression $<a_1(-), a_2(-), \ldots, a_n(-)>$ for
the exponential (with $A=\{a_1,\dots a_n\}$). As an example, take the
functor $\D(X) =
2\times X^A$ of deterministic automata. The expressions corresponding
to this functor would then be the set of closed and guarded
expressions given by the following BNF:
\[
\E::= x \mid \mu x.\E \mid <0, <a_1(\E), a_2(\E), \ldots, a_n(\E)>>
\mid <1,<a_1(\E), a_2(\E), \ldots, a_n(\E)>>
\]
This syntax can be perceived as an explicit and complete description
of the automaton. This means that underspecification is nonexistent and
the compactness of regular expressions is lost. As an example of the
verbosity present in this new language, take $A=\{a,b,c\}$ and
consider the language that accepts words with only $a$'s and has at
last one $a$ (described by $aa^*$ in Kleene's regular expressions). In
the language $\Exp_\D$ it would be written as $\mu x. a(l<1>\oplus x)$.
Using the approach described above it would be encoded as the
expression
\[
\mu x.  <0,
<a(<1,<a(x),b(\mathit{empty}),c(\mathit{empty})>>),b(\mathit{empty}),c(\mathit{empty})>>
\]
where $\mathit{empty}= \mu y. <0,<a(y), b(y), c(y)>$ is the expression
denoting the empty language. The approach we presented before, by
allowing underspecification, provides
a more user-friendly syntax and stays close to the know syntaxes for
deterministic automata and LTSs. 

In what follows we will formally present a language for polynomial
coalgebras. We start by introducing the definition of polynomial functor, which we
take from~\cite{adamek06}. 

\begin{definition}[Polynomial Functor]
Sums of the Cartesian power functors are called polynomial functors:
\[
\Pol_\Sigma(X) = \coprod_{\sigma \in \Sigma} X^{\mathit{ar}(\sigma)}
\]
Here, $\coprod$ stands for ordinary coproduct and 
 the indexing set $\Sigma$ is a signature, that is a possibly infinite
collection of symbols $\sigma$, each of 
which is equipped with a finite cardinal $\mathit{ar}(\sigma)$, called
the arity of $\sigma$. 
\end{definition}

\begin{definition}[Expressions and axioms for polynomial functors]
Let $\Pol_\Sigma$ be a polynomial functor. The set $\Exp_{\Pol_\Sigma}$ of expressions for ${\Pol_\Sigma}$ is given by the closed and guarded expressions generated by the following BNF, where $\sigma\in \Sigma$ and $x\in V$, for $V$ a set of fixed point variables:
\[
\E_i ::= x \mid \mu x. \E \mid \sigma(\E_1, \ldots, \E_{\mathit{ar}(\sigma)})
\]
accompanied by the equations:
\begin{align*}
&\gamma[\mu x.\gamma/x] \equiv \mu x . \gamma && \mathit{(FP)}  \\
&{\gamma[\E/x]  \equiv  \E \Rightarrow \mu
x.\gamma \equiv \E} &&\mathit{(Unique)}
\\
&\E_1\equiv \E_2 \Rightarrow \E[\E_1/x]\equiv \E[\E_2/x], \text{\ \ \ \  if
$x$ is free in $\E$}&& (\mathit{Cong})
\\
&\mu x.\gamma \equiv \mu y. \gamma[y/x],\ \ \   \text{ if
$y$ is not free in $\gamma$}& &\mathit {(\alpha-\mathit{equiv})} 
\end{align*}

\end{definition}
Providing the set $\Exp_{\Pol_\Sigma}$ with a coalgebraic structure is
 achieved using induction on the number of unguarded occurrences of nested fixed points:
\[
\begin{array}{l}
\delta \colon \Exp_{\Pol_\Sigma}\to \coprod\limits_{\sigma \in \Sigma} (\Exp_{\Pol_\Sigma})^{\mathit{ar}(\sigma)}\\
\delta(\mu x.\E) = \delta(\E[\mu x.\E /x])\\
\delta(\sigma(\E_1,\ldots,\E_{\mathit{ar}(\sigma)})) = \kappa_\sigma(<\E_1,\ldots,\E_{\mathit{ar}(\sigma)}>) 
\end{array}
\]
We are now ready to state and prove Kleene's theorem.
\begin{theorem}[Kleene's theorem for polynomial functors]
Let $\Pol_\Sigma$ be a polynomial functor.

\begin{enumerate}[\em(1)]
\item  For every locally finite coalgebra $(S,g\colon S\to \Pol_\Sigma(S))$ and for every $s\in S$ there exists an expression  $\E\in \Exp_{\Pol_\Sigma}$ such that $\E\sim s$.
\item For every expression $\E\in \Exp_{\Pol_\Sigma}$ there is a finite  coalgebra $(S, g\colon S\to \Pol_\Sigma(S))$ with $s\in S$ such that $s\sim \E$.
\end{enumerate}
\end{theorem}
\begin{proof}
Point $1.$ amounts to solve a system of equations. Let $<s> = \{s_1,\ldots, s_n\}$. We associate with each $s_i \in <s>$ an expression $\expr{s_i} = A^n_i$, where $A^n_i$ is defined inductively as in the proof of~\ref{kleene1}, with $A^{k+1}_i = A^k_i\{A^k_{k+1}/x_{k+1}\}$ and $A^0_i = A_i$ given by
\[
A_i = \mu x_{s_i}. \sigma(x_{s_1'}, \ldots, x_{s_{\mathit{ar}(\sigma)}'}), \ \ \ g(s_i) = \kappa_\sigma(s_{1}', \ldots, s_{\mathit{ar}(\sigma)}')
\]
It remains to prove that $s_i \sim \expr{s_i}$, for all $s_i\in <s>$. We observe that
\[
R = \{ <s_i, \expr{s_i}> \mid s_i \in <s> \}
\]
is a bisimulation, since, for $g(s_i) = \kappa_\sigma(s_{1}', \ldots, s_{\mathit{ar}(\sigma)}')$, we have
\[
\begin{array}{@{}ll}
\ \ \ \; \delta(\expr{s_i}) \\
= \delta((\mu x_i.  \sigma(x_{s_1'}, \ldots, x_{s_{\mathit{ar}(\sigma)}'}))  \{A^0_{1}/x_{1}\} \ldots \{A^{n-1}_{n}/x_{n}\}) \\[1.5ex]
=\delta(\mu x_i.  \sigma(x_{s_1'}, \ldots, x_{s_{\mathit{ar}(\sigma)}'})  \{A^0_{1}/x_{1}\} \ldots  \{A^{i-2}_{i-1}/x_{i-1}\}  \{A^i_{i+1}/x_{i+1}\} \ldots\{A^{n-1}_{n}/x_{n}\})\\[1.5ex]
=\delta( \sigma(x_{s_1'}, \ldots, x_{s_{\mathit{ar}(\sigma)}'})  \{A^0_{1}/x_{1}\} \ldots  \{A^{i-2}_{i-1}/x_{i-1}\}  \{A^i_{i+1}/x_{i+1}\} \ldots\{A^{n-1}_{n}/x_{n}\}[A^n_i/x_i])&  \\[1.5ex]
=\delta( \sigma(x_{s_1'}, \ldots, x_{s_{\mathit{ar}(\sigma)}'})  \{A^0_{1}/x_{1}\} \ldots  \{A^{i-2}_{i-1}/x_{i-1}\}  \{A^i_{i+1}/x_{i+1}\} \ldots\{A^{n-1}_{n}/x_{n}\}\{A^n_i/x_i\})&
\\[1.5ex]
=\delta( \sigma(x_{s_1'}, \ldots, x_{s_{\mathit{ar}(\sigma)}'})\{A^0_{1}/x_{1}\} \ldots  \{A^{i-2}_{i-1}/x_{i-1}\} \{A^n_i/x_i\} \{A^i_{i+1}/x_{i+1}\} \ldots\{A^{n-1}_{n}/x_{n}\})&\\[1.5ex]
=\kappa_\sigma(\expr{s_1'},\ldots, \expr{s_{\mathit{ar}(\sigma)}'})
\end{array}
\]
For point $2$, we observe that the subcoalgebra $<\E>$, for any $\E\in \Exp_{\Pol_\Sigma}$ is finite, since the set $\mathit{cl}(\E)$ containing all sub-formulas and unfoldings of fixed points of $\E$, which is finite, is a subcoalgebra of  $(\Exp_{\Pol_\Sigma},\delta)$. The fact that in this point, contrary to what happened in Theorem~\ref{kleene2}, we do not need to quotient the set of expressions is a direct consequence of the absence of underspecification or, more concretely, of the expressions $\emp$ and $\oplus$.
\end{proof}

The proof of soundness and completeness would follow a similar strategy as in the previous section and we will omit it here.

In order to be able to compare the language introduced in this section
with the language obtained in our previous approach, we have to define
an infinitary version of the operator $\myplus$ and extend the
framework accordingly. We start by defining the aforementioned operator on sets:
$\mybigplus_{i\in I} X_i = \left(\coprod_{i\in I} X_i\right) \cup \{\bot,\top\} $
and the corresponding functor, for which we shall use the same symbol, is defined pointwise in the same way as for $\myplus$.
Note that $\myplus$ is a special case of this operator (resp. functor)
for $I$ a two element set. In fact, for simplicity, we shall only
consider this operator for index sets $I$ with  two or more elements. 

There is a natural transformation between polynomial functors and
the class of non-deterministic functors extended with $\mybigplus$:
every polynomial functor $\Pol_\Sigma$ is mapped to
$$\overline{\Pol}_\Sigma(X) = \mybigplus_{\sigma\in \Sigma} \
X^{\mathit{ar}(\sigma)}$$ 
Next, we slightly alter the definition of expressions. Instead of the
expressions $l[-]$ and $r[-]$ we had before for $\myplus$ we add
an expression $i[-]$ for each $i\in I$ and the expected typing rule:
\[
\rules{\vdash \E\colon \F_j\lhd \G\ \ \ \ \ \ j\in I}
      {\vdash j[\E] \colon  \mybigplus_{i\in I} \F_i\lhd \G}
\]
All the other elements in our story are adjusted in the expected
way. We show what happens in the axiomatization.
For $\myplus$ we had the rules 
\[
\begin{array}{l}
 l[\E_1 \oplus \E_2] \equiv l[\E_1] \oplus l[\E_2] \ \ \ \ \ 
   r[\E_1 \oplus \E_2] \equiv r[\E_1] \oplus
r[\E_2]\ \ \ \ \ \ 
 l[\E_1] \oplus r[\E_2] \equiv l[\emp] \oplus r[\emp]
 \end{array}
\]
which are now replaced by
\[
\begin{array}{l}
i[\E_1]\oplus i[\E_2] \equiv i[\E_1\oplus \E_2]\ \ \ \ \ \ 
i[\E_1]\oplus j[\E_2] \equiv k[\emp]\oplus l[\emp], \ \ \ \
i\neq j, k\neq l
\end{array}
\]

It is natural to ask what is the relation between the sets of
expressions $\Exp_{\Pol_\Sigma}$ and $\Exp_{\overline{\Pol}_\Sigma}$.
The set  $\Exp_{\Pol_\Sigma}$ is bijective to the subset of
$\Exp_{\overline{\Pol}_\Sigma}$ containing only fully specified
expressions, that is expressions $\E$ for which the subcoalgebra
$<\E>$ does not contain any state for which
$\delta_{\overline{\Pol}_\Sigma}$ evaluates to $\bot$ and $\top$.
This condition is purely semantical and we were not able to find a purely syntactic restriction that would capture it.

We next repeat the exercise above for finitary functors. 
A finitary functor $\Fin$ is a functor that is a quotient of a
polynomial functor, {\em i.e.} there exists a natural transformation
$\eta\colon \Pol_\Sigma \to \Fin$, whose components $\eta_X\colon \Pol_\Sigma(X) \to
\Fin(X)$ are epimorphisms.  We define $\Exp_{\Fin} =
\Exp_{\Pol_\Sigma}$. 
\begin{theorem}[Kleene's theorem for finitary functors]
Let $\Fin$ be a finitary functor.
\begin{enumerate}[\em(1)]
\item Let $(S,f)$ be a
locally-finite $\Fin$-coalgebra. Then, for any $s \in S$, there exists an
expression $\expr s \in \Exp_{\Fin}$ such that $s\sim
\expr s$. 
\item Let $\E\in \Exp_{\Fin}$. Then, there exists a finite
$\Fin$-coalgebra $(S,f)$ with $s\in S$ such that $s\sim \E$. 
\end{enumerate} 
\end{theorem}
\proof
Let $\Fin$ be a finitary functor (quotient of a polynomial
functor $\Pol_\Sigma$). 

\mycirc{1} Let $(S,f)$ be a locally finite $\Fin$-coalgebra and
let $s\in S$. We denote by $T = \{s_1,\ldots, s_n\}$ the state space
of the subcoalgebra $<s>$ (with $s_1=s$). We then have that there
exists an $f^\sharp$ making the following diagram commute:
\[
\xymatrix{T \ar[d]_{f^\sharp}\ar[r]^{\mathit{id}}& T\ar[d]^f\ar[r] & S\ar[d]^f\\
\Pol_\Sigma(T)\ar@{->>}[r]_{\eta_S} & \Fin(T)\ar[r]& \Fin(S)}
\]
We then build $\expr{s}$ w.r.t $f^\sharp$ just as in
Theorem~\ref{kleene1} (note that $(T,f^\sharp)$ is finite) and 
the result follows because $\expr s \sim_{\Fin} s \Leftarrow \expr s
\sim_{\Pol_\Sigma} s$ (consequence of naturality).

\mycirc{2} Let $\E\in \Exp_{\Fin}$. By Theorem~\ref{kleene2},
there
exists a finite $\Pol_\Sigma$-coalgebra $(S, f)$ with $s\in S$ such that
$s\sim_{\Pol_\Sigma}
\E$. Thus, we take $(S, \eta_S\circ f)$ and we have a finite
$\Fin$-coalgebra with $s\in S$ such that $\E\sim_{\Fin} s$.
\qed

For the axiomatization a bit more ingenuity is required. 
One needs to derive which extra axioms are induced by the epimorphism
and
then prove that they are sound and complete. 

For instance, the finite powerset functor (which we included in the
syntax of non-deterministic functors) 
is the classical example of a finitary functor. It is the quotient of
the polynomial functor $\Pol_\Sigma(X)
= 1+X+X^2+\ldots$ (this represents lists of length $n$) by identifying
lists that contain precisely the same elements (that is, eliminating
repeated elements and abstracting from the ordering).

The syntax for $\Exp_{\Pol_\Sigma}$ is the set of
closed and guarded expressions given by the following BNF:
\[
\E ::= x \mid \mu x.\E \mid i(\E_1,\ldots,\E_i), \ \ \ i\in \mathbb{N}\\
\]
together with the axioms for the fixed point,
$(\alpha-\mathit{equiv})$ and $\mathit{(Cong)}$. 

Taking into account the restriction mentioned we would have to include
the extra axioms:
\[
\begin{array}{l}
i(\E_1,\ldots,\E_i) \equiv i(\E_1',\ldots,\E_i') \text { if }
\{\E_1,\ldots \E_i\} =  \{\E_1',\ldots \E_i'\}\\
i(\E_1,\E_2,\ldots,\E_i) \equiv (i-1)(\E_1,\E_3,\ldots,\E_i) \text { if }
\E_1\equiv \E_2\\

\end{array}
\]
In this case, one can see that this set of axioms is sound and
complete, by simply proving, for $\Pol_\Sigma(X)
= 1+X+X^2+\ldots$, $\Exp_{\Pol_\Sigma}/_\equiv \cong \Exp_\pow/_\equiv$
(since we already had a language and sound and complete axiomatization
for the $\pow$ functor). The restricted syntax and axioms needs to 
be derived for each concrete finitary functor. 
Finding a uniform way of defining such restricted syntax/axioms and
also uniformly proving soundness and completeness is a challenging problem and it is left as future work. 

\section{Discussion}\label{sec:conclusions5}

We presented a systematic way of deriving, from the type of a system,
a language of (generalized) regular expressions and a sound and
complete axiomatization thereof. We showed the analogue of Kleene's
theorem, proving the correspondence of the behaviours captured by
the expressions and the systems under consideration. The whole
approach was illustrated with five examples: deterministic finite
automata, partial deterministic automata, non\hyph deterministic automata,
labelled transition systems and automata on guarded strings. Moreover,
all the results presented in~\cite{BRS08} for Mealy machines can be recovered as
a particular instance of the present framework. 
 
Iterative theories have been introduced by Elgot~\cite{elgot} as a model of
computation and they formalize potentially infinite computations as
solutions of recursive equations. The main example of an iterative theory is the theory of
regular trees, that is trees which have on finitely many distinct
subtrees. Ad\'amek, Milius and Velebil have presented Elgot's work
from a coalgebraic perspective~\cite{AMV03,adamekmscs2006}, simplified
some of his original proofs, and generalized the notion of free
iterative theory to any finitary endofunctor of every locally
presentable category. The language modulo the axioms we will associate
with each functor is closely related to the work above: it is an 
initial iterative algebra. This also shows the connection of our work with the
work by Bloom and \'Esik on iterative algebras/theories~\cite{BE93}.
 It would be
interesting to investigate the connections with iterative
algebras further. 

In~\cite{jacobs06}, a bialgebraic review of deterministic automata and
regular expressions was presented. One of the main results of~\cite{jacobs06}
was a description of the free algebra and Brzozowski coalgebra
structure on regular expressions as a bialgebra with respect to a GSOS
law. We expect that this extends to our framework, but fully working this out is left as future
work. 

In this paper we studied coalgebras for $\mathbf{Set}$ functors. It is
an important and challenging question to extend our results to other
categories. Following our work, S. Milius~\cite{milius:lics10} has 
showed how to derive a language and sound and complete
axiomatization for the functor $\mathbb R \times \id$ in the category
of vector spaces and linear maps. It would also be interesting to study functors
over metric spaces~\cite{TuriR98,BreugelW06}.

In his seminal paper~\cite{kleene}, S.~Kleene introduced an
algebraic description of regular languages: regular expressions. This
was the precursor of many papers, including this one. Salomaa \cite{salomaa} presented a sound and complete axiomatization
for proving the equivalence of regular expressions. This was later
refined by Kozen in \cite{kozen}: he showed that Salomaa's
axiomatization is non-algebraic, in the sense that it is unsound under
substitution of alphabet symbols by arbitrary regular expressions, and
presented an algebraic axiomatization. In~\cite{milner}, Milner
introduced a set of expressions for finite LTS's and proved an
analogue of Kleene's theorem: each expression denotes the behaviour
of a finite LTS and, conversely, the behaviour of a finite LTS can
be specified by an expression. He also provided an axiomatization for his expressions, with the property that two expressions are provably equivalent if and only if they are bisimilar.

Our approach is inspired by the work of Kleene, Kozen and Milner. For that
reason, we have $\emp$ and $\oplus$ in the syntax of our
expressions, which allow to have underspecification and
overspecification. These features had to be reflected in the type of
the coalgebras we are able to deal with: the class of functors
considered include join-semilattices as constant functors and
$\myplus$ instead of the ordinary coproduct, which has allowed us to
remain in the category $\mathbf{Set}$. The fact that underspecification and
overspecification can be captured by a semilattice structure, plus the
fact that the axiomatization provides the set of expressions with a
join semilattice structure, hint (as one of the reviewers pointed out)
that the whole framework could have been studied directly in the
category of join-semilattices. This is indeed true, but, for
simplicity, we decided to remain in the category $\mathbf{Set}$. It is
not clear how much could be gained by directly working on join
semi-lattices. 

The connection between regular expressions and coalgebras was first
explored in~\cite{Rut98c}. There deterministic  automata, the set of formal languages
and regular expressions are all presented as coalgebras of the
functor $2 \times \id^A$ (where $A$ is the alphabet, and $2$ is the
two element set). It is then shown that the standard semantics of
language acceptance of automata and the assignment of languages to
regular expressions both arise as the unique homomorphism into the
final coalgebra of formal languages. The coalgebra structure on the
set of regular expressions is determined by their so-called {\em
Brzozowski} derivatives~\cite{Brz64}. In the present paper, the set
of expressions for the functor $\D(S) = 2 \times S^A$ differs from
the classical definition in that we do not have Kleene star and full
concatenation (sequential composition) but, instead, the least fixed
point operator and action prefixing. Modulo that difference, the
definition of a coalgebra structure on the set of expressions in
both~\cite{Rut98c} and the present paper is essentially the same.
All in all, one can therefore say that standard regular expressions
and their treatment in \cite{Rut98c} can be viewed as a special
instance of the present approach. This is also the case for the
generalization of the results in~\cite{Rut98c} to automata on guarded
strings~\cite{kozen08}. Finally, the present paper extends the results
in our FoSSaCS'08 paper~\cite{BRS08},
where a sound and complete specification language and a synthesis
algorithm for Mealy machines is given. Mealy machines are coalgebras
of the functor $(\B \times \id)^A$, where $A$ is a finite input
alphabet and $\B$ is a finite semilattice for the output alphabet. Part of the material of the present paper is based on two conference papers: our FoSSaCS'09 paper~\cite{regexp} and our LICS'09 paper~\cite{BRS09b}.

In the last few years, several proposals of specification languages
for coalgebras
appeared \cite{Moss99,Rossiger00,Jacobs01,Gol02,CirsteaP04,Bonsangue-Kurz05,Bonsangue-Kurz06,SP07,KV07}.
Our approach is similar in
spirit to that of~\cite{Gol02,Rossiger00,Jacobs01,SP07} in that we use the
ingredients of a functor for typing expressions, and differs
from~\cite{Rossiger00,Jacobs01} because we do not need an explicit "next-state"
operator, as we can deduce it from the type information. The modal operators associated to a functor
in~\cite{Rossiger00,Jacobs01,SP07} can easily be related with the
expressions considered in our language. As an example, consider the
expression $<\pi_2>[\kappa_1]<\alpha> \bot$, written in the syntax
of~\cite{Rossiger00}, which belongs to the language associated with the
functor $2\times (\id + 1)$ (the modal operator $<\alpha>$ is next
operator associated with the identity functor). In our language, this
would be represented by $r<l[\emp]>$.

Apart from~\cite{KV07}, the languages mentioned above do not include fixed point
operators. Our language of regular expressions can be seen as an
extension of the coalgebraic logic
of~\cite{Bonsangue-Kurz05} with fixed point operators, as well as the multi-sorted logics of~\cite{SP07}, and it is similar to a fragment of the
logic presented in~\cite{KV07}. However, our goal is rather
different: we want (1) a finitary language that characterizes
exactly all \emph{locally finite} coalgebras; (2) a Kleene like
theorem for the language or, in other words, a map (and not a
relation) from expressions to coalgebras and vice-versa. Similar to
many of the works above, we also derive a modular
axiomatization, sound and complete with respect to observational
equivalence. From the perspective of modal logic, the second half of
Kleene's theorem, where we show how to construct a coalgebra from an
expression, is the same as constructing a canonical model.
In~\cite{SP07}, the models presented for the multi-sorted logics are
multi-sorted coalgebras, whereas here we remain in the world of
coalgebras in the category  \textbf{Set}, constructing, from an
expression in $\Exp_\G$, for a given functor $\G$, a $\G$-coalgebra.
 Further exploring the connections with the approach presented
in~\cite{SP07} is a promising research path, opening the door to extending our framework for more general classes of functors.   

In conclusion, we mention a recent generalization of the present
approach: all the results presented in this paper can be extended in order to accommodate systems with \emph{quantities}, such as probability or costs~\cite{BBRS09}. The main
technical challenge
is that quantitative systems have an inherently non-idempotent
behaviour and thus the proof of Kleene's theorem and the
axiomatization require extra care. This extension allows for the derivation of specification languages and axiomatizations for a wide variety of systems, which include weighted automata, simple probabilistic systems (also known as Markov chains) and systems with mixed probability and non-determinism (such as Segala systems). For instance, we have derived a language and an axiomatization for the so-called stratified systems. The language is equivalent to the one presented in~\cite{GSS95}, but no axiomatization was known. 

The derivation of the syntax and axioms associated with each non-deterministic functor has been implemented in the coinductive prover CIRC~\cite{circ}. This allows for automatic reasoning about the equivalence of expressions specifying  systems. 

\paragraph{\textbf{Acknowledgements.}} The authors are grateful for useful
comments from several people: Filippo Bonchi, Helle Hansen, Bartek
Klin, Dexter Kozen, Clemens Kupke, Stefan Milius, Prakash Panagaden,
Ana Sokolova, Yde Venema and Erik de Vink.  The title of this paper
was inspired by the title of a section of a paper of Dexter
Kozen~\cite{kozen08}. The proof of soundness and completeness was
simplified (when compared with the one presented in our LICS
paper~\cite{BRS09b}) inspired by recent work of Stefan Milius on
expressions for linear systems (personal communication). Finally, we
would like to thank the three anonymous reviewers for their very detailed
reports, which greatly improved the presentation of the paper.\vspace{-24 pt}

\bibliographystyle{plain}
\bibliography{refs}
\vspace{-20 pt}
\end{document}